\documentclass[a4paper,11pt]{amsart}

\usepackage{amsmath}
\usepackage{amsthm}
\usepackage{booktabs}
\usepackage{amssymb}
\usepackage{amsbsy}
\usepackage{amsopn}
\usepackage{amsfonts}
\usepackage{amstext}
\usepackage{bbm}

\usepackage{amscd}
\usepackage{tikz}
\usepackage{graphicx}
\usepackage{verbatim}
\usepackage{bm}
\usepackage{commath,mathtools,setspace}
\usepackage{paralist}
\usepackage{extarrows}
\usepackage{color}
\usepackage{stmaryrd}
\usepackage{mathrsfs}
\usepackage{stackengine}
\usepackage{comment}
\usepackage{stmaryrd}  
\usepackage{dsfont}
\usepackage[T1]{fontenc}
\usepackage{graphicx}
\usepackage{blindtext} 
\usepackage{geometry}
\usepackage[colorlinks=true,linkcolor=blue,citecolor=blue]{hyperref}
\usepackage[colorinlistoftodos,prependcaption,textsize=small]{todonotes}

\definecolor{red}{RGB}{255,0,0}
\definecolor{green}{RGB}{0,100,0}
\definecolor{blue}{RGB}{0,0,255}

\numberwithin{equation}{section}

\newtheorem{theorem}{Theorem}[section]
\newtheorem{lemma}[theorem]{Lemma}

\newtheorem{proposition}[theorem]{Proposition}

\newtheorem{remark}[theorem]{Remark}
\newtheorem{definition}[theorem]{Definition}
\newtheorem{example}[theorem]{Example}

\newcommand{\vect}[1]{\underline{#1}}
\newcommand{\one}{\underline{\textbf{1}}}
\newcommand{\N}{\vect{N}}
\newcommand{\n}{{\vect{n}}}
\renewcommand{\k}{\vect{k}}
\newcommand{\K}{\vect{K}}

\newcommand{\X}{\vect{X}}
\newcommand{\Y}{\vect{Y}}

\newcommand{\ul}{\vect{\ell}}
\renewcommand{\L}{\vect{L}}
\newcommand{\x}{\vect{x}}

\newcommand{\y}{\vect{y}}

\newcommand{\ualpha}{\vect{\alpha}}

\newcommand{\R}{\mathbb{R}}
\newcommand{\Natural}{\mathbb{N}}

\newcommand{\Ccal}{\mathcal{C}}
\newcommand{\Scal}{\mathcal{S}}
\newcommand{\Acal}{\mathcal{C}'}
\newcommand{\Gcal}{\mathcal{G}}

\DeclareMathOperator{\mult}{mult}
\DeclareMathOperator{\nm}{nm}
\DeclareMathOperator{\hg}{hg}
\DeclareMathOperator{\nhg}{nhg}
\DeclareMathOperator{\Dir}{Dir}
\DeclareMathOperator{\IDir}{IDir}
\DeclareMathOperator{\DirMult}{DirMult}
\DeclareMathOperator{\DirNm}{DirNm}

\DeclareMathOperator{\E}{\mathbb{E}}
\renewcommand{\P}{\mathbb{P}}

\def\nb{{\mathfrak{nb}}}
\def\onb{\overline{\mathfrak{nb}}}
\def\pa{{\mathfrak{pa}}}

\newcommand{\gbinom}[3]{\binom{#1}{#2}_{\!\!#3}}
\newcommand{\gbinomm}[3]{\left[\genfrac{}{}{0pt}{}{#1}{#2}\right]_{\!#3}}

\makeatletter
\@namedef{subjclassname@2020}{%
	\textup{2020} Mathematics Subject Classification}
\makeatother


\begin{document}

\title{Graphical Models for Multivariate Count Data}
\author{Iza Danielewska}
\email{iza.danielewska.dokt@pw.edu.pl}

\author[B. Ko{\char"8A}odziejek]{Bartosz Ko{\char"8A}odziejek}
\email{bartosz.kolodziejek@pw.edu.pl}
\address{Faculty of Mathematics and Information Sciences, Warsaw University of Technology, Koszykowa 75, \mbox{00-662} Warsaw, Poland}

\thanks{This research was funded in part by National Science Centre, Poland, UMO-2022/45/B/ST1/00545.
\\ For the purpose of Open Access, the authors have applied a CC-BY public copyright licence to any Author Accepted Manuscript (AAM) version arising from this submission.}

\begin{abstract}
The classical multinomial, negative multinomial, hypergeometric, and negative hypergeometric distributions are naturally organized by two features of the sampling scheme: sampling with or without replacement and stopping after a fixed number of draws or a fixed number of failures. We complete the graphical analogue of this scheme for decomposable graphs by adding graphical hypergeometric and graphical negative hypergeometric distributions to the previously introduced graphical multinomial and graphical negative multinomial models in Danielewska et al. (2025). The resulting four families provide a unified parametric framework for graphical modeling of multivariate count data, in which dependence and admissible configurations are encoded by a graph. They interpolate between products of univariate distributions for the empty graph and the corresponding classical multivariate distributions for the complete graph, while retaining explicit Markov factorizations and tractable sampling representations. We further develop a unified Bayesian hierarchy based on graphical Dirichlet-type distributions, obtaining explicit posterior and predictive laws. The framework is particularly natural for count data arising under exclusion or incompatibility constraints. We discuss several such applications and illustrate its practical potential using Rydberg-atom excitation data.
\end{abstract}

\maketitle

\section{Introduction}

Multivariate count data arise in a wide range of applications, and a substantial literature is devoted to constructing parametric families capable of representing dependence between counts while retaining interpretable distributional structure. Classical examples include the multinomial, negative multinomial, multivariate hypergeometric, and multivariate negative hypergeometric distributions \cite{JohnsonKotzBalakrishnan1997}. A central challenge is to enrich such classical families with structured dependence while preserving explicit probabilistic interpretations, tractable marginal and conditional distributions, and feasible statistical inference.

Graphs provide a natural way of representing such dependence through conditional independence. For discrete data, the connection between conditional independence and log-linear models goes back to \cite{DarrochLauritzenSpeed1980}; see also \cite{Lauritzen} for the general theory of graphical models. Decomposable graphs are particularly attractive because Markov distributions admit clique--separator factorizations, allowing the joint distribution to be reconstructed from compatible clique marginals. This structure underlies classical analysis of discrete graphical models \cite{DarrochLauritzenSpeed1980} and Bayesian analysis through hyper-Markov laws and conjugate priors \cite{Dawid1993,MassamLiuDobra2009}. 

Closer to the present work are constructions of explicit multivariate count distributions with structured dependence. Peyhardi and Fernique \cite{PeyhardiFernique2017} characterize the graphical structure of convolution splitting distributions, while Peyhardi, Fernique and Durand \cite{PeyhardiFerniqueDurand2021} develop a broader splitting framework encompassing many classical multivariate count distributions, including the multinomial, negative multinomial, multivariate hypergeometric and negative hypergeometric families. More recently, Tree Pólya splitting distributions \cite{ValiquetteEtAl2026} introduce hierarchical splitting structures allowing richer dependence patterns.

The framework considered here combines graphical conditional independence with familiar parametric count distributions. In previous work \cite{GDir}, the graphical multinomial and graphical negative multinomial distributions, denoted by $\mult_G$ and $\nm_G$, were introduced for decomposable graphs. These distributions are natural graphical counterparts of their classical analogues in a strong sense. The characterizations in \cite[Section~3.4]{GDir} show that, for connected decomposable graphs, imposing the global Markov property together with compatible multinomial, respectively negative multinomial, clique marginals characterizes the corresponding graphical family. Thus $\mult_G$ and $\nm_G$  arise naturally from requiring classical marginal distributions on the cliques together with the conditional independences encoded by $G$.

The present paper completes and develops the resulting family of four graphical count distributions. Recall that the classical multinomial, negative multinomial, hypergeometric, and negative hypergeometric distributions can be organized according to two features of their sampling schemes: whether sampling is with or without replacement, and whether the experiment is observed for a fixed number of draws or until a fixed number of failures. The graphical multinomial and graphical negative multinomial distributions introduced in \cite{GDir} occupy the two with-replacement positions. We introduce the two remaining members, the graphical hypergeometric distribution $\hg_G$ and the graphical negative hypergeometric distribution $\nhg_G$, thereby obtaining the graphical counterpart of the classical scheme summarized in Table~\ref{tab:sampling_schemes}.

\begin{table}[!b]
    \centering
    \caption{Discrete distributions classified by sampling scheme.}
    \label{tab:sampling_schemes}
    \begin{tabular}{lll}
        \toprule
        \textbf{Stopping rule}
        & \textbf{With replacement}
        & \textbf{Without replacement} \\
        \midrule
        Fixed number of draws
        & Multinomial
        & Hypergeometric \\
        Fixed number of failures
        & Negative multinomial
        & Negative hypergeometric \\
        \bottomrule
    \end{tabular}
\end{table}

A core contribution of this paper is to give this four-family scheme a common and transparent probabilistic interpretation. The connection with the Cartier--Foata normal form \cite{cartier1969commutation,DiekertMetivier1997} identifies admissible sequences of cliques of $G^\ast$
 (equivalently, independent sets of $G$) as the natural combinatorial objects underlying the graphical constructions. In this representation, $\mult_G$ describes counts obtained from a fixed number of clique-valued draws with replacement, whereas $\nm_G$ describes the corresponding failure-stopped mechanism under the specified Cartier--Foata transition probabilities. The interpretation of $\mult_G$ used here is the same as that developed in \cite{GDir}, and the failure-stopped construction of $\nm_G$ likewise goes back to \cite[Section~4]{GDir}. The contribution here is to place these constructions in a common clique-sampling framework, making their sampling interpretation and their relationship to one another more transparent, and to extend this framework to the two without-replacement models. Specifically, the new models $\hg_G$ and $\nhg_G$ provide the corresponding finite-content mechanisms, with $\nhg_G$ giving the without-replacement analogue of the failure-stopped experiment. In this sense, the present interpretation refines and unifies the constructions of \cite{GDir} and embeds them in a single probabilistic scheme encompassing all four graphical count distributions.

This four-family viewpoint is also important for statistical applications. The four distributions should not be regarded as models for four unrelated types of data. Rather, they describe different observation schemes or statistical questions for the same underlying graph-constrained mechanism. For example, $\mult_G$ provides an appropriate model for the aggregate
counts arising from repeated independent sampling from the clique law
underlying $\mult_G(1,\y)$. If a collection of such draws is conditioned
on its aggregate count vector, then $\hg_G$ describes the corresponding
finite-content problem when only a subset of the positions is observed.
Thus an application supporting the graphical multinomial mechanism
naturally gives rise to a graphical hypergeometric problem when the
sampling design or conditioning information changes. The same
relationship connects $\nm_G$ and $\nhg_G$ in the corresponding
failure-stopped experiments.

For decomposable graphs, the four families are constructed through graph-multinomial coefficients associated with the multivariate independence polynomial \cite{levit2005independence}. Their probability mass functions admit explicit clique--separator and DAG factorizations and consequently satisfy the global Markov property, while their clique marginals belong to the corresponding classical families. At the two extremes, the empty graph yields products of the appropriate univariate distributions and the complete graph yields the corresponding classical multivariate distributions. Decomposability is an essential part of the general construction: outside the decomposable setting, the analogous graph-multinomial expressions do not in general satisfy the required Markov property, and we therefore do not define the four families in general for non-decomposable graphs. One notable special case is the graphical Bernoulli distribution
$\mult_G(1,\y)$, which is well defined for an arbitrary graph and
satisfies the corresponding graphical Markov property; see Remark~\ref{rem:non-decomposable}. The graphical Bernoulli distribution coincides with the hard-core model \cite{Scott_2005}, a classical model from statistical physics describing configurations subject to a hard exclusion constraint.

The two new distributions also preserve the classical conditioning relationships between sampling with and without replacement. The graphical hypergeometric distribution $\hg_G$ is obtained by conditioning two independent graphical multinomial random vectors on their sum, while $\nhg_G$ arises analogously from graphical negative multinomial random vectors. Together with clique--separator and DAG factorizations, these representations yield explicit marginal and conditional distributions and explicit sampling constructions.

Our approach is related to other recent parametric constructions for multivariate counts, but differs in the role assigned to graphical structure. The splitting models of \cite{PeyhardiFerniqueDurand2021} construct multivariate count distributions by combining a distribution for a total count with an allocation of that total among coordinates, while the Tree Pólya splitting models of \cite{ValiquetteEtAl2026} recursively apply such allocations along a partition tree. These constructions share with the present work an interest in explicit and tractable families built from classical multivariate count distributions. They do not, however, generally produce families indexed by a prescribed nontrivial conditional-independence graph. For the positive convolution splitting distributions covered by
\cite{PeyhardiFernique2017}, the resulting minimal graphical models are
either complete or empty. Tree Pólya splitting generates substantially richer dependence, but its partition tree represents a hierarchical allocation scheme rather than a conditional-independence graph; indeed, \cite{ValiquetteEtAl2026} identifies determination of the associated probabilistic graphical models as a direction for further research. 

The four models also fit naturally into a common Bayesian framework. Conjugate analysis has a rich history for decomposable graphical models, beginning with hyper-Markov laws and hyper-Dirichlet priors \cite{Dawid1993}; see also \cite{MassamLiuDobra2009}. In the present parametric setting, the graphical Dirichlet and graphical inverted Dirichlet distributions introduced in \cite{GDir} provide the corresponding mixing distributions. They generate two parallel Bayesian hierarchies connecting $\mult_G$ with $\hg_G$ and $\nm_G$ with $\nhg_G$. Integrating out the graphical parameters produces the $G$-Dirichlet multinomial and $G$-Dirichlet negative multinomial laws, while posterior distributions of the latent parameters and predictive distributions of the unobserved remainder remain available explicitly. Table~\ref{tab:bayesian_conjugacy_schemes} summarizes these two parallel hierarchies.

The graphical construction is particularly natural when elementary configurations are governed by exclusion or incompatibility relations. An admissible configuration is then an independent set of $G$, equivalently a clique of the complement graph $G^\ast$. Such structures arise naturally in hard-core systems, communication networks, mutually exclusive biological events, loss and resource-sharing systems, and polymer models
\cite{jiang2010csma,ciriello2012memo,kelly_loss_networks,fernandez_procacci}.
The four graphical distributions provide different count data models for such systems according to the sampling scheme and the inferential question of interest.

We illustrate this connection using experimental excitation patterns of Rydberg atoms subject to blockade constraints \cite{kim2022rydberg,kim2022data}. Under the idealized blockade constraint, an admissible excitation pattern
is an independent set of the interaction graph. We therefore use
$\mult_G(1,\y)$ as a parametric model for the distribution of such
blockade-consistent patterns. Since this is precisely the graphical Bernoulli case, the model remains well defined even when the interaction graph is non-decomposable. Experimental imperfections also generate configurations outside this theoretical support; our empirical analysis consequently concerns the conditional distribution of the blockade-consistent observations. For the graph configurations considered, the resulting goodness-of-fit analysis shows that the graphical multinomial model provides an adequate description of the admissible excitation patterns. Although this particular analysis uses $\mult_G$, it illustrates the graph-constrained mechanism underlying the four-family construction. In particular, when the interaction graph is decomposable, $\hg_G$ provides the corresponding finite-population model when the inferential question concerns a subset of a fixed collection of admissible excitation patterns conditional on their aggregate counts. Thus $\mult_G$ and $\hg_G$ address different observation schemes for the same underlying graphical system.

The paper is organized as follows.
Section~\ref{sec:preliminaries} introduces the necessary graph-theoretic notation, graph polynomials, and graph-multinomial coefficients.
Section~\ref{sec:mult-nm} recalls the graphical multinomial and graphical negative multinomial distributions.
Section~\ref{sec:hypergeometric_models} introduces their graphical hypergeometric and graphical negative hypergeometric counterparts and establishes their main distributional and graphical properties.
Section~\ref{sec:interpretation} develops probabilistic interpretations of the four models based on the Cartier--Foata normal form.
Section~\ref{sec:bayesian_inference} presents the corresponding Bayesian hierarchies and posterior and predictive distributions.
Section~\ref{sec:application_rydberg} gives the Rydberg-atom application, and Section~\ref{sec:independent_set_examples} discusses further settings in which the four graphical count models may be useful.
Proofs are collected in Appendix~\ref{app:proofs}.

\section{Preliminaries}
\label{sec:preliminaries}

We collect the notation and graph-theoretic concepts used throughout the paper,
largely following \cite{GDir}. Let \(G=(V,E)\) be a finite simple undirected
graph with nonempty vertex set \(V\).

\subsection{Notation and graph-theoretic notions}

Let \(\Natural=\{0,1,\ldots\}\) and
\(\Natural_+=\{1,2,\ldots\}\). Vectors are underlined. For
\(A\subseteq V\) and \(\underline{v}\in\mathbb{R}^V\), write
\[
\underline{v}_A=(v_i)_{i\in A},
\qquad
|\underline{v}|=\sum_{i\in V}v_i.
\]
For \(A\subseteq V\), let \(\one_A\in\{0,1\}^V\) denote its indicator
vector, whose \(v\)-th coordinate equals \(1\) if \(v\in A\) and \(0\)
otherwise. For \(a>0\) and \(\vect{b}\in\Natural^A\), we use the convention
\[
\binom{a}{\vect{b}}
=
\frac{\Gamma(a+1)}
{\Gamma(a-|\vect{b}|+1)\prod_{i\in A}b_i!},
\]
whenever the right-hand side is well defined. For
\(\x\in(0,\infty)^V\) and \(\y\in\mathbb{R}^V\), write
\[
\x^{\y}=\prod_{v\in V}x_v^{y_v}.
\]

We write \(i\sim j\) whenever \(\{i,j\}\in E\). The open and closed
neighbourhoods of \(i\in V\) are, respectively,
\[
\nb_G(i)=\{j\in V\colon i\sim j\},
\qquad
\onb_G(i)=\nb_G(i)\cup\{i\}.
\]
For \(A\subseteq V\), let \(G_A=(A,E_A)\) denote the subgraph induced by
\(A\). A set \(C\subseteq V\) is a clique if \(G_C\) is complete and is
maximal if it is not contained in a larger clique. We denote by
\(\Ccal_G\) the set of all cliques of \(G\), including \(\emptyset\), and
by \(\Ccal_G^+\) the set of maximal cliques.

The graph \(G\) is decomposable, or chordal, if it has no induced cycle
of length at least four. In this case, its maximal cliques admit an
ordering \((C_1,\ldots,C_K)\) satisfying the running intersection
property: for each \(k=2,\ldots,K\), there exists \(j<k\) such that
\[
S_k
:=
C_k\cap\bigcup_{i=1}^{k-1}C_i
\subseteq C_j.
\]
The sets \(S_k\) are called separators. We write
\[
\Scal_G^-=\left\{S_k \colon k=2,\ldots,K,\quad S_k\neq\emptyset
\right\},
\qquad
\nu_S
=
\bigl|\{k\in\{2,\ldots,K\}\colon S_k=S\}\bigr|,
\quad S\in\Scal_G^-,
\]
for the set of distinct separators and their multiplicities.

A distribution of a random vector \(\N=(N_v)_{v\in V}\) satisfies the
global Markov property with respect to \(G\) if
\[
\N_A\perp\!\!\!\perp\N_B\mid\N_C
\]
whenever \(A,B,C\subseteq V\) are disjoint and \(C\) separates \(A\)
and \(B\) in \(G\). For a decomposable graph \(G\), any PMF admitting the clique--separator
factorization
\begin{align}
\label{eq:clique-separator-factorization}
\P(\N=\n)
=
\frac{
\displaystyle\prod_{C\in\Ccal_G^+}
\P(\N_C=\n_C)
}{
\displaystyle\prod_{S\in\Scal_G^-}
\P(\N_S=\n_S)^{\nu_S}
}
\end{align}
whenever the right-hand side is well defined, satisfies the global Markov
property with respect to \(G\). Conversely, for a strictly positive PMF on a
product state space, the global Markov property is equivalent to
\eqref{eq:clique-separator-factorization}, see \cite{Lauritzen}.

Let \(\mathcal G\) be a directed acyclic graph (DAG) on \(V\). Its skeleton is the
undirected graph obtained by replacing every directed edge by an undirected
edge. For \(v\in V\), let
\[
\pa(v)
=
\{u\in V:\ u\to v\text{ in }\mathcal G\}
\]
denote the set of parents of \(v\). We call \(\mathcal G\) moral (or perfect) if
\(\operatorname{pa}(v)\) is a clique in its skeleton for every \(v\in V\).

\subsection{Graph polynomials and coefficients}

Let \(G^*=(V,E^*)\) be the complement of \(G\), where $
E^*
=
\bigl\{
\{i,j\}\colon i,j\in V,\ i\neq j,\ \{i,j\}\notin E
\bigr\}$.

\begin{definition}[Graph polynomials]
The graph polynomials \(\delta_G,\Delta_G\colon\mathbb{R}^V\to\mathbb{R}\)
are defined by
\[
\delta_G(\y)
=
\sum_{C\in\Ccal_{G^*}}\prod_{v\in C}y_v,
\qquad
\Delta_G(\x)
=
\sum_{C\in\Ccal_{G^*}}(-1)^{|C|}
\prod_{v\in C}x_v.
\]
Thus, \(\Delta_G(\x)=\delta_G(-\x)\).
\end{definition}

The polynomial \(\delta_G\) is the multivariate independence polynomial
of \(G\), equivalently the clique polynomial of \(G^*\); see
\cite{goldwurm1998clique,levit2005independence,Mult24}.

\begin{definition}[Graph-multinomial coefficients, {\cite{GDir}}]
Let \(G=(V,E)\) be a finite simple graph. For \(r\in\Natural_+\), the
graph-multinomial coefficients of the first type are defined by
\[
\delta_G(\y)^r
=
\sum_{\n\in\Natural^V}
\gbinom{r}{\n}{G}\y^\n.
\]
For \(s>0\), the graph-multinomial coefficients of the second type are
defined by the formal power series expansion
\[
\Delta_G(\x)^{-s}
=
\sum_{\n\in\Natural^V}
\gbinomm{|\n|+s-1}{\n}{G}\x^\n.
\]
\end{definition}

For decomposable graphs, both families admit explicit product
factorizations.

\begin{lemma}[{\cite[Lemma 2.15]{GDir}}]
\label{lem:coeff_factorization}
If \(G\) is decomposable, then, for \(r\in\Natural_+\), \(s>0\), and
\(\n\in\Natural^V\),
\[
\gbinom{r}{\n}{G}
=
\frac{
    \displaystyle\prod_{C\in\Ccal_G^+}
    \binom{r}{\n_C}
}{
    \displaystyle\prod_{S\in\Scal_G^-}
    \binom{r}{\n_S}^{\nu_S}
},
\qquad
\gbinomm{|\n|+s-1}{\n}{G}
=
\frac{
    \displaystyle\prod_{C\in\Ccal_G^+}
    \binom{|\n_C|+s-1}{\n_C}
}{
    \displaystyle\prod_{S\in\Scal_G^-}
    \binom{|\n_S|+s-1}{\n_S}^{\nu_S}
}.
\]
\end{lemma}

\section{Graphical multinomial and negative multinomial distributions}
\label{sec:mult-nm}

We briefly recall the two discrete models defined in \cite{GDir} and their
conjugate priors.

\begin{definition}[Graphical multinomial and negative multinomial distributions]
Let $G=(V,E)$ be a decomposable graph.
\begin{enumerate}
    \item A random vector $\K$ has the $G$-multinomial distribution,
    denoted by $\K \sim \mult_G(r,\y)$, with parameters
    $r \in \Natural_+$ and $\y \in (0,\infty)^V$, if
    \[
    \P(\K=\k)
    =
    \gbinom{r}{\k}{G} \y^{\k} \,\delta_G(\y)^{-r},
    \qquad \k \in \mathbb{N}_{G,r},
    \]
    where
    \[
    \mathbb{N}_{G,r}
    :=
    \{\n \in \Natural^V \colon \max_{C \in \Ccal_G^+} |\n_C| \le r \}.
    \]

    \item A random vector $\K$ has the $G$-negative multinomial distribution,
    denoted by $\K \sim \nm_G(r,\x)$, with parameters
    $r>0$ and 
    \begin{align}\label{eq:MG}
\x\in M_G :=
\left\{
\x \in (0,\infty)^V \colon
\Delta_{G_A}(\x_A) > 0 \ \text{for all } A \subseteq V
\right\},
\end{align}
if
    \[
    \P(\K=\k)
    =
    \gbinomm{|\k|+r-1}{\k}{G}
    \x^{\k} \Delta_G(\x)^r,
    \qquad \k \in \Natural^V.
    \]
\end{enumerate}
\end{definition}

The parameter domain \(M_G\) of $\nm_G$ can be characterized as the largest subset of
\((0,\infty)^V\) on which the series
\[
    \sum_{\k\in\Natural^V}
    \gbinomm{|\k|+r-1}{\k}{G}
    \x^{\k}
\]
converges, \cite[Lemma 2.18]{GDir}. 
It is useful to keep in mind the simple bounds
\[
    \left\{\x\in(0,\infty)^V\colon \sum_{i\in V}x_i<1\right\}
    \subset M_G
    \subset
    (0,1)^V .
\]

Both distributions $\mult_G$ and $\nm_G$ are globally Markov with respect to $G$ when $G$ is
decomposable; see \cite[Section~3.3]{GDir}.  For the multinomial model there is
one exceptional case in which decomposability is not needed.  We record it
separately.

\begin{remark}[Bernoulli distribution for arbitrary graphs]
\label{rem:mult_one_arbitrary_G}
Let \(G=(V,E)\) be any finite simple graph and let \(\y\in(0,\infty)^V\).  Since
\[
\delta_G(\y) = \sum_{C\in \Ccal_{G^\ast}} \prod_{v\in C} y_v
=
\sum_{\k\in\mathbb N_{G,1}} \y^{\k},
\qquad\mbox{
we have}\qquad 
\gbinom{1}{\k}{G}
=
\mathbf 1_{{\k\in\mathbb N_{G,1}}},
\qquad \k\in\Natural^V.
\]
By definition,
\[
    \mathbb N_{G,1}
    =
    \{\k\in\{0,1\}^V\colon  k_i+k_j\le 1
      \text{ for every } \{i,j\}\in E\},
\]
and thus, we obtain 
\[
    \P(\K=\k)
    =
    \frac{1}{\delta_G(\y)}
    \prod_{v\in V} y_v^{k_v}
    \prod_{\{i,j\}\in E} \mathbf 1_{\{k_i+k_j\le 1\}},
    \qquad \k\in\{0,1\}^V .
\]
This factorization immediately gives the pairwise Markov property. Because the distribution has structural zeros, the usual equivalence between
pairwise and global Markov properties for strictly positive laws on product
spaces cannot be invoked directly.

Let \(A,B,S\subseteq V\) be disjoint and suppose that \(S\)
separates \(A\) and \(B\) in \(G\).  Fix \(s\in\{0,1\}^S\) with
\(\P(\K_S=\underline{s})>0\), and let \(D_1,\ldots,D_m\) be the connected components of
\(G_{V\setminus S}\).  Conditional on \(\K_S=\underline{s}\), the law of
\(\K_{V\setminus S}\) is proportional to
\[
    \prod_{\ell=1}^m
    \left[
        \prod_{v\in D_\ell} y_v^{k_v}
        \prod_{\{i,j\}\in E\colon  i,j\in D_\ell}
            \mathbf 1_{\{k_i+k_j\le 1\}}
        \prod_{\{i,j\}\in E\colon  i\in D_\ell,\ j\in S}
            \mathbf 1_{\{k_i+s_j\le 1\}}
    \right].
\]
Thus \(\K_{D_1},\ldots,\K_{D_m}\) are conditionally independent given
\(\K_S=\underline{s}\).  Since \(S\) separates \(A\) and \(B\), no component \(D_\ell\)
intersects both \(A\) and \(B\), and therefore
\(
    \K_A \perp\!\!\!\perp \K_B \mid \K_S .
\)
Hence \(\mult_G(1,\y)\) is globally Markov with respect to \(G\) for arbitrary
finite \(G\).

It is also worth mentioning that the graphical Bernoulli distribution can be viewed as the push-forward of the hard-core measure \cite{Scott_2005} under the map that records the vertices belonging to a sampled independent set, the hard-core model being a classical model from statistical physics, see also \cite[Section 4.1]{GDir}.
\end{remark}

\begin{remark}[Non-decomposable graphs]
\label{rem:non-decomposable}

The four families of distributions based on coefficients considered in this paper are restricted to decomposable graphs, except for the case \(\mult_G(1,\y)\) described in Remark \ref{rem:mult_one_arbitrary_G}, which is well defined and globally Markov for every finite graph \(G\). For a decomposable graph, the global Markov property together with compatible
clique marginals determines the joint distribution through the clique-separator
factorization.

The graph-multinomial coefficient formulas themselves remain well defined
for arbitrary graphs. However, except for special cases such as the Bernoulli case described in
Remark~\ref{rem:mult_one_arbitrary_G}, the resulting distributions need not satisfy the
global Markov property with respect to \(G\). In particular, for the cycle \(C_4\), the
coefficient-based graphical multinomial distribution with \(r=2\) has multinomial edge
marginals but is not \(C_4\)-Markov, see \cite[Remark~3.14]{GDir}.

For non-decomposable graphs, a natural alternative is to define graphical
models through prescribed classical clique marginals together with the global
Markov property, rather than directly through graph-multinomial coefficients.
This leads to a clique-marginal completion problem, in which one seeks a joint
distribution compatible with the specified local laws. The development of these
completion-based models, including questions of compatibility, existence,
uniqueness, computation, and extensions to infinite-support distributions, is
left for future work.
\end{remark}

\section{Graphical hypergeometric models}
\label{sec:hypergeometric_models}

We now introduce the two new families of distributions that are the focus of this paper, corresponding to sampling without replacement.

\subsection{The graphical hypergeometric distribution}

\begin{definition}[Graphical hypergeometric distribution]
\label{def:hg_G}
Let $G=(V,E)$ be a decomposable graph. A random vector $\N$ has the $G$-hypergeometric distribution with parameters $(M, \K, r)\in \Natural_+\times \Natural^V\times \Natural_+$ satisfying  
\(
M \geq \max_{C \in \Ccal_G^+} |\K_C| \text{ and } M>r
\)
if its PMF is
\begin{align}\label{PMF_hg}
\P(\N=\n) = \frac{\gbinom{r}{\n}{G} \gbinom{M-r}{\K-\n}{G}}{\gbinom{M}{\K}{G}}=:\hg_G(M,\K,r)(\n), \quad \n \in S^{\hg_G}_{M,\K,r}, \end{align}
where the support is
\begin{align*}
S^{\hg_G}_{M,\K,r} =\left\{\n \in \Natural^V \colon n_v \le K_v \text{ for all } v \in V, \,\,0 \leq r-|\n_C|\leq M-|\K_C| \text{ for all }C\in\Ccal_G\right\}.
\end{align*}
We write $\N \sim \hg_G(M, \K, r)$.
\end{definition}

The graphical hypergeometric distribution arises naturally as the conditional distribution of a $\mult_G$ vector, given the sum of it and another independent $\mult_G$ vector. This parallels the classical relationship between the binomial and hypergeometric distributions.

\begin{proposition}[Connection to $\mult_G$]\label{prop:hg}
Let $M,r\in\Natural_+$ with $M>r$, and let $\N_1 \sim \mult_G(r, \y)$ and $\N_2 \sim \mult_G(M-r, \y)$ be independent random vectors. Then the conditional distribution of $\N_1$ given their sum is $G$-hypergeometric:
\[ 
\N_1 \mid \N_1 + \N_2 = \K \sim \hg_G(M,\K,r).
\]

\end{proposition}
The proof follows from a direct calculation using the PMFs and the convolution property of graph-multinomial coefficients.

Like the other graphical models, the $\hg_G$ distribution possesses the global Markov property and its clique marginals are classical hypergeometric distributions.

\begin{theorem}\label{thm:properties_hg}
Assume $\N \sim \hg_G(M, \K, r)$ for a decomposable graph $G$.
\begin{enumerate}
    \item The distribution of $\N$ has the global Markov property with respect to $G$.
    \item For any clique $C \in \Ccal_G$, the marginal distribution of $\N_C$ is the classical multivariate hypergeometric distribution, $\N_C \sim \hg(M, \K_C, r)$.
\end{enumerate}
\end{theorem}
The proof of these properties relies on the factorization of the graph-multinomial coefficients given in Lemma \ref{lem:coeff_factorization}.

\begin{theorem}\label{thm:dag-hg}
    Assume that \(G\) is a decomposable graph. Let \(\Gcal\) be a moral DAG with skeleton \(G\). If \(\N\sim \hg_G( M,\K, r)\), then 
    \begin{align} \label{eq-Markov_prop}
    \P(\N=\n)=\prod_{v\in V} \hg\left( M - |\K_{\pa(v)}|, K_v,r - |\n_{\pa(v)}| \right)(n_v), \quad \n \in S^{\hg_G}_{M,\K,r},
    \end{align}
    where \(\hg\left(M, K, r\right)\) is the PMF for the classical hypergeometric distribution, i.e.,
\[
\hg\left( M,K_v, r\right)(n_v) = \frac{\binom{K_v}{n_v}\binom{M-K_v}{r-n_v}}{\binom{M}{r}}.
\]
    In particular, for every \(v\in V\) and every
\(\n_{\pa(v)}\) such that
\(
\P(\N_{\pa(v)}=\n_{\pa(v)})>0,
\)
we obtain
\[
    \P(N_v=n_v| \N_{\pa(v)}=\n_{\pa(v)})=\hg\left( M - |\K_{\pa(v)}|,K_v, r - |\n_{\pa(v)}| \right)(n_v).
\]
\end{theorem}

\subsection{The graphical negative hypergeometric distribution}

We now introduce the graphical negative hypergeometric distribution, which provides the fourth member of the family of graphical count distributions.

\begin{definition}[Graphical negative hypergeometric distribution]
\label{def:ngh_G}
Let $G=(V,E)$ be a decomposable graph. A random vector $\N$ has the $G$-negative hypergeometric distribution with parameters $(M, \K, r)\in \Natural_+\times \Natural^V\times \R_+$ satisfying 
\(
M \geq r
\)
if its PMF is
\[ 
\P(\N=\n) = \frac{\gbinomm{|\n|+r-1}{\n}{G} \gbinomm{M-r+|\K|-|\n|}{\K-\n}{G}}{\gbinomm{M+|\K|}{\K}{G}}, \quad \n \in S^{\nhg_G}_{M,\K,r}, \]
where the support is 
\[
S^{\nhg_G}_{M,\K,r} = \{\n \in \Natural^V \colon 0 \le n_v \le K_v \text{ for all } v \in V \}.
\]
We write $\N \sim \nhg_G(M, \K, r)$.
\end{definition}

Here \(r>0\) is not required to be an integer. For
\(r\in\{1,\ldots,M\}\), the distribution admits the failure-stopped
sampling interpretation developed in Section~\ref{sec:clique-nhg},
whereas noninteger admissible values of \(r\) provide an analytic extension.

We provide results analogous to those for the graphical hypergeometric distribution below. 

\begin{proposition}[Connection to $\nm_G$]\label{prop:nhg}
Let \(M\geq r\), $\N_1 \sim \nm_G(r, \x)$ and $\N_2 \sim \nm_G(M-r+1, \x)$ be independent random vectors. Then the conditional distribution of $\N_1$ given their sum is $G$-negative hypergeometric:
\[ \N_1 \mid \N_1 + \N_2 = \K \sim \nhg_G(M,\K,r). \]
\end{proposition}

\begin{theorem}\label{thm:properties_nhg}
Let $\N \sim \nhg_G(M, \K, r)$ on a decomposable graph $G$.
\begin{enumerate}
    \item The distribution of $\N$ has the global Markov property with respect to $G$.
    \item For any clique $C \in \Ccal_G$, the marginal distribution of $\N_C$ is the classical multivariate negative hypergeometric distribution, $\N_C \sim \nhg(M, \K_C, r)$.
\end{enumerate}
\end{theorem}
The properties are established using similar arguments as for the $\hg_G$ model, leveraging the factorization properties of the second-type graph-multinomial coefficients.

\begin{theorem}\label{thm:dag-nhg}
    Assume that \(G\) is a decomposable graph. Let \(\Gcal\) be a moral DAG with skeleton \(G\). If \(\N\sim \nhg_G( M,\K, r)\), then 
    \begin{align} \label{eq-Markov_prop-nhg}
    \P(\N=\n)=\prod_{v\in V} \nhg\left( M + |\K_{\pa(v)}|, K_v,r + |\n_{\pa(v)}| \right)(n_v), \quad \n \in S^{\nhg_G}_{M,\K,r}
    \end{align}
    where \(\nhg\left(M, K_v, r\right)(n_v)\) is the PMF for the classical negative hypergeometric distribution, i.e.,
\[
\nhg\left( M,K_v, r\right)(n_v) = \frac{\binom{n_v+r-1}{n_v}\binom{M-r+K_v-n_v}{K_v-n_v}}{\binom{M+K_v}{K_v}}.
\]
    In particular, for every \(v\in V\) and every
\(\n_{\pa(v)}\) such that
\(
\P(\N_{\pa(v)}=\n_{\pa(v)})>0,
\)
we obtain
\[
    \P(N_v=n_v| \N_{\pa(v)}=\n_{\pa(v)})=\nhg\left( M + |\K_{\pa(v)}|,K_v, r + |\n_{\pa(v)}| \right)(n_v).
\]
\end{theorem}

\section{Probabilistic interpretation of the four graphical models}\label{sec:interpretation}

Before turning to the individual constructions, let us emphasize why such
probabilistic interpretations are useful. The formulas defining the four
graphical distributions are compact algebraic expressions involving the graph
polynomials \(\delta_G\) and \(\Delta_G\) and the corresponding
graph-multinomial coefficients. A probabilistic construction explains what kind
of random mechanism can produce these formulas. In particular, it shows that the
models are not only formal analogues of the classical multinomial, negative
multinomial, hypergeometric, and negative hypergeometric laws, but also arise
from natural sampling schemes in which the admissible outcomes are constrained
by a graph.

The common feature of these constructions is that observations are built from
independent sets of \(G\), equivalently from cliques of the complement graph
\(G^\ast\). The multinomial and negative multinomial models correspond to
sampling with replacement, while the hypergeometric and negative hypergeometric
models correspond to the analogous sampling schemes without replacement, or to
conditioning on fixed aggregate contents. This point of view is important for applications: in these sections we
discuss several settings in which independent-set-valued configurations arise
naturally: feasible schedules in wireless networks, mutually
exclusive alterations in cancer genomics, pairwise loss and resource-sharing
systems, and abstract polymer models.

\subsection{\texorpdfstring{Probabilistic interpretation of \(\mult_G\)}
{Probabilistic interpretation of multG}}

The classical multinomial distribution arises from repeated independent sampling with replacement from a finite collection of categories. In the graphical setting, a single draw is replaced by a graph-constrained admissible configuration: a clique of the complement graph \(G^*\), or equivalently an independent set of \(G\). The distribution \(\mult_G(1,\y)\) describes one such random configuration, while \(\mult_G(r,\y)\) describes the aggregate counts obtained from \(r\) independent configurations.

Let \(G=(V,E)\) be a finite simple graph and let \(G^*=(V,E^*)\) denote its complement.
For a clique \(C\in\Ccal_{G^*}\), let \(\one_C\in\{0,1\}^V\) denote its indicator vector. As discussed in Remark \ref{rem:mult_one_arbitrary_G}, the following construction of $\mult_G(1,\y)$ holds for arbitrary graphs, not necessarily decomposable.

\begin{theorem}[Theorem 4.3 in \cite{GDir}]
\label{thm-mult-proba-interpretation}
	For \(\y\in (0,\infty)^V\), let \(\ell\) be a random variable valued in \(\mathcal{C}_{G^*}\), with distribution
	\begin{align}\label{eq:distell}
    \mathbb{P}(\ell=C)=\frac{\prod_{v\in C}y_v}{\delta_G(\y)},\qquad  C\in \mathcal{C}_{G^*}.
    \end{align}
	Then 
    $\one_\ell\sim \mult_G(1,\y)$.
\end{theorem}	
If $G$ is decomposable, then by \cite[Corollary 3.6]{GDir}, for every $r\in\mathbb N_+$,
\[
\sum_{j=1}^r \mathbf 1_{\ell_j}
\sim
\mult_G(r,\y),
\]
where $(\ell_1,\ldots,\ell_r)$ are independent and identically distributed according to \eqref{eq:distell}. Thus, in the decomposable case, $\mult_G(r,\y)$ admits the following probabilistic interpretation: it is the distribution of the vertex-count vector obtained by sampling $r$ independent cliques of the complement graph $G^*$, with clique $C$ chosen with weight proportional to $\prod_{v\in C}y_v$.

It is also worth recalling the direct connection with the hard-core model from
statistical physics \cite{Scott_2005}; see also \cite[Section~4.1]{GDir}. Given activities
\(\underline{\lambda}=(\lambda_v)_{v\in V}\in(0,\infty)^V\), the hard-core model on
\(G\) is the probability distribution on the family \(\mathcal I(G)\) of
independent sets of \(G\), where a set is independent if no two of its vertices
are adjacent, defined by
\[
    \mathbb P(I)
    =
    \frac{1}{Z_G(\underline{\lambda})}
    \prod_{v\in I}\lambda_v,
\qquad 
    Z_G(\underline{\lambda})
    =
    \sum_{I\in\mathcal I(G)}
    \prod_{v\in I}\lambda_v .
\]
Since independent sets of \(G\) are precisely cliques of the complement graph
\(G^\ast\), we have $Z_G(\underline{\lambda})=\delta_G(\underline{\lambda})$.
Consequently the indicator vector \(\one_I\) of a hard-core sample satisfies $\one_I\sim \mult_G(1,\underline{\lambda})$.

\subsection{\texorpdfstring{Probabilistic interpretation of \(\hg_G\)}
{Probabilistic interpretation of hgG}}

A first interpretation of \(\hg_G\) follows from Proposition~\ref{prop:hg}: it is the conditional distribution of one graphical multinomial count vector given the sum of two independent graphical multinomial count vectors. However, that construction involves an auxiliary parameter vector \(\y\), whereas the distribution \(\hg_G(M,\K,r)\) itself does not depend on \(\y\). We now give a parameter-free sampling interpretation.
\begin{theorem}\label{thm:interp_hg}
Let \(G\) be a decomposable graph. Fix \(M\in\Natural_+\), \(r\in\{1,\dots,M-1\}\), and
\(\K\in\mathbb N^V\) such that
\(
M\ge \max_{C\in\Ccal_G^+}|\K_C|.
\)
Consider the set of ordered \(M\)-tuples of cliques of \(G^*\) with total content \(\K\):
\[
\mathcal B_M(\K)
=
\left\{
(C^{(1)},\dots,C^{(M)})\in \Ccal_{G^*}^M:
\sum_{j=1}^M \one_{C^{(j)}}=\K
\right\}.
\]
Choose an element of \(\mathcal B_M(\K)\) uniformly at random and define $\N=\sum_{j=1}^r \one_{C^{(j)}}$.
Then
\[
\N\sim \hg_G(M,\K,r).
\]
\end{theorem}

\subsection{\texorpdfstring{Probabilistic interpretation of \(\nm_G\)}
{Probabilistic interpretation of nmG}}\label{sec:nmg_int}
We first recall the probabilistic interpretation of the \(\nm_G\) model from \cite{GDir} and we then provide a new interpretation.

We fix a finite decomposable graph \(G=(V,E)\). We regard $V$ as an alphabet, and refer to its elements as letters. Let $\mathfrak V$ be the free monoid generated by $V$ under concatenation; thus $\mathfrak V$ is the set of all finite words over the alphabet~$V$.
The complement graph \(G^*=(V,E^*)\) determines a partial commutation relation on
the alphabet \(V\): two distinct letters \(i_1,i_2\in V\) are allowed to commute
whenever \(\{i_1,i_2\}\in E^*\). Let \(\equiv\) be the smallest monoid congruence on \(\mathfrak V\) such that
\[
    w_1 i_1 i_2 w_2
    \equiv
    w_1 i_2 i_1 w_2
\]
for all \(w_1,w_2\in\mathfrak V\) and all \(\{i_1,i_2\}\in E^*\). We define
\[
    L:=\mathfrak V/\equiv .
\]
We say that \(L\) is the free  quotient monoid induced by \(G=(V,E)\). Its elements are denoted by \([w]\), where \(w\in\mathfrak V\) is any representative of the
equivalence class.

The following representation of the \(\nm_G\) distribution was obtained in
\cite[Theorem~4.2]{GDir}: for \(\x\in M_G\), let \(\ell\) be an \(L\)-valued
random variable with distribution
\[
    \mathbb P(\ell=[w])
    =
    \frac{\prod_{j=1}^{|w|} x_{w_j} }
    {\sum_{[w']\in L} \prod_{i=1}^{|w'|} x_{w'_i}},
    \qquad [w]\in L,
\]
where \(w=w_1\cdots w_{|w|}\in\mathfrak{V}\) is any representative of the class \([w]\), and \(|w|\) denotes the length of \(w\). The expression is well defined because commuting adjacent letters does not change the product.
Let \(\varepsilon([w])\in\Natural^V\) record the number of occurrences of
each letter in \([w]\in L\). Then
\[
    \varepsilon(\ell)\sim\nm_G(1,\x).
\]
This result is based on the Cartier--Foata theorem \cite{cartier1969commutation}.
We note, however, that this construction does not directly exhibit a
``sampling until failure'' scheme, as one might expect for a negative
multinomial distribution.


We now introduce a ``sample cliques until failure'' representation of \(\nm_G(1,\x)\). In the free quotient monoid $L$ associated with \(G\), two letters commute if and only if they are adjacent in \(G^*\). Thus the cliques of \(G^*\) are precisely the sets of pairwise commuting letters.

Let
\[
    \Acal_{G^*}:=\Ccal_{G^*}\setminus\{\emptyset\}.
\]
We use the standard Cartier--Foata admissibility relation on cliques: for \(C,D\in\Acal_{G^*}\), write
\(C\to D\) if
\begin{equation}
\label{eq:CF-adm}
    \forall\,v\in D\ \exists\,u\in C
    \quad\text{such that}\quad
    u=v \ \text{or}\ \{u,v\}\in E.
\end{equation}
In other words, each vertex in the next clique \(D\) must either already belong to the current clique \(C\), or be adjacent in \(G\) to at least one vertex of \(C\).

\begin{lemma}[{\cite[Section 2.3]{DiekertMetivier1997}}]
\label{lem:CF-normal-form}
Let \(L\) be the free quotient monoid induced by $G$. 
Then every nonempty element \(\ell\in L\) admits a unique factorization
\[
    \ell=C_1C_2\cdots C_T,
    \qquad
    C_t\in\Acal_{G^\ast},\quad C_t\to C_{t+1},\quad t=1,\dots,T-1.
\]
\end{lemma}

Thus sampling from $L$ may equivalently be viewed as sampling a finite admissible sequence of nonempty cliques of \(G^*\).

For \(C\subseteq V\), write
\[
    \onb_G(C)
    :=
    C\cup\{v\in V\colon \exists u\in C\text{ such that }\{u,v\}\in E\}.
\]
In particular, $\onb_G(\emptyset)=\emptyset$.

Define for $C\in\Ccal_{G^*}$,
\begin{align}\label{eq:px}
    \Phi_{\x}(C)
    :=
    \Delta_{G_{V\setminus \onb_G(C)}}
    \bigl(\x_{V\setminus \onb_G(C)}\bigr),
    \qquad 
    p_{\x}(C):=\Phi_{\x}(C)\, \x^{\one_C}.
\end{align}
Note that $p_{\x}(\emptyset) = \Delta_G(\x)$. 
We have  
\[
V\setminus \onb_G(C) = 
    \{v\in V\setminus C\colon \{u,v\}\in E^\ast
      \text{ for every }u\in C\} = \bigcap_{u\in C}\nb_{G^*}(u). 
\]

Recall the definition of the set $M_G$ from \eqref{eq:MG}. 
\begin{lemma}\label{lem:propP} Let \(G=(V,E)\) be a decomposable graph and let \(\x\in M_G\). Then $p_{\x}$ is a probability distribution on $\Ccal_{G^\ast}$, i.e., 
    \[
   \sum_{C\in\Ccal_{G^*}}p_{\x}(C) = 1\qquad\mbox{and} \qquad p_{\x}(C)>0 \mbox{ for }C\in\Ccal_{G^*}.  
    \]
    Moreover,
\[
    \Phi_{\x}(C)
    =
    p_{\x}(\emptyset)
    +
    \sum_{D\colon C\to D}p_{\x}(D),
    \qquad C\in\Acal_{G^*}.
\]
\end{lemma}

\begin{example}
Let \(G\) be the chain \(1-2-3\). Then \(G^*\) has the single edge \(\{1,3\}\), and hence $\Ccal_{G^*}
    =
    \{\emptyset,\{1\},\{2\},\{3\},\{1,3\}\}$. 
Then,
\[
    \Delta_G(\x)=1-x_1-x_2-x_3+x_1x_3
\]
and the values of \(p_{\x}(C)\) are
\[
\begin{array}{c|ccccc}
C
& \emptyset
& \{1\}
& \{2\}
& \{3\}
& \{1,3\}
\\ \hline
p_{\x}(C)
& \Delta_G(\x)
& x_1(1-x_3)
& x_2
& x_3(1-x_1)
& x_1x_3
\end{array}
\]
\end{example}

We construct a random finite admissible clique sequence as a Markov chain
stopped upon hitting the empty clique.  First sample $C_0\in \Ccal_{G^*}$ according to 
\[
    \mathbb P(C_0=C)=p_{\x}(C),
    \qquad C\in \Ccal_{G^*}.
\]
If \(C_0=\emptyset\), set \(T=0\) and output the empty sequence.
Otherwise, given \(C_t=C\in\Acal_{G^*}\), sample \(C_{t+1}\) according to
\[
    \mathbb P(C_{t+1}=D\mid C_t=C)
    =
    \begin{cases}
        \dfrac{p_{\x}(D)}{\Phi_{\x}(C)},
        & D\in \{\emptyset\}\cup
        \{D\in\Acal_{G^*}: C\to D\}, \\[1.2em]
        0, & \text{otherwise}.
    \end{cases}
\]
Finally, define the hitting time
\(
    T:=\inf\{t\in\Natural \colon  C_t=\emptyset\}.
\)
The resulting admissible clique sequence is
\(
    (C_0,\ldots,C_{T-1}),
\)
with the convention that this sequence is empty when \(T=0\).

\begin{theorem}
\label{prop:clique-chain-nm1}
Let $G$ be a decomposable graph. 
We have 
\[
    \K:=\sum_{t=0}^{T-1}\one_{C_t}\sim\nm_G(1,\x),
\]
with the convention that \(\K= \underline{0}\) if \(T=0\). 
\end{theorem}

If $G$ is decomposable, then by \cite[Corollary 3.6]{GDir}, for every $r\in\mathbb N_+$,
\[
\sum_{j=1}^r \K_j
\sim
\nm_G(r,\x),
\]
where $\K_1,\ldots,\K_r$ are i.i.d. $\nm_G(1,\x)$.

\subsection{\texorpdfstring{Probabilistic interpretation of \(\nhg_G\)}
{Probabilistic interpretation of nhgG}}
\label{sec:clique-nhg}

A first interpretation of \(\nhg_G\) follows from its conditional construction in
terms of graphical negative multinomial random vectors in Proposition \ref{prop:nhg}. However, that construction
involves an auxiliary parameter vector \(\x\), whereas the distribution
\(\nhg_G(M,\K,r)\) itself does not depend on \(\x\). We now give a parameter-free
sampling interpretation.

Let \(\Acal_{G^*}\) and the admissibility relation \(\to\) be as in the previous subsection. A finite sequence
\[
    \mathbf C=(C_1,\dots,C_T),
    \qquad C_t\in\Acal_{G^*},
\]
is called admissible if $C_t\to C_{t+1}$ for all $t=1,\dots,T-1$.
We also allow the empty admissible sequence, denoted by \(()\). For an admissible
sequence \(\mathbf C\), define its content by
\[
    \varepsilon(\mathbf C)
    :=
    \sum_{t=1}^T \one_{C_t}\in\Natural^V,
\]
with the convention that \(\varepsilon(())=\underline 0\).

For \(q\in\Natural_+\) and \(\K\in\Natural^V\), consider the set of ordered
\(q\)-tuples of admissible sequences with total content \(\K\):
\[
    \mathcal A_q(\K)
    :=
    \left\{
    (\mathbf C^{(1)},\dots,\mathbf C^{(q)}):
    \sum_{j=1}^q \varepsilon(\mathbf C^{(j)})=\K
    \right\}.
\]

\begin{theorem}
\label{thm:interp_nhg}
Let \(G\) be a decomposable graph. Fix \(M\in\Natural_+\), \(r\in\{1,\dots,M\}\), and \(\K\in\Natural^V\). Choose an
element of \(\mathcal A_{M+1}(\K)\) uniformly at random,
\[
    (\mathbf C^{(1)},\dots,\mathbf C^{(M+1)})
    \sim
    \mathrm{Unif}\bigl(\mathcal A_{M+1}(\K)\bigr),
\]
and define $\N
    =
    \sum_{j=1}^r \varepsilon(\mathbf C^{(j)})$. Then
\[
    \N\sim\nhg_G(M,\K,r).
\]
\end{theorem}

\section{Bayesian hierarchies for graphical hypergeometric models}
\label{sec:bayesian_inference}

We now develop a Bayesian hierarchical framework for the graphical
hypergeometric and graphical negative hypergeometric models.
We first recall the graphical Dirichlet-type priors introduced in \cite{GDir}, which will serve as mixing distributions in the hierarchical construction.  
\begin{definition}[{\cite[Section 5]{GDir}}]	\label{def-graph-dir-inverted-dir}
	Let \(G=(V,E)\) be a finite decomposable graph.
	\begin{enumerate}
		\item A probability measure on \(\mathbb{R}^V\) is the \(G\)-Dirichlet distribution with parameters \((\ualpha,\beta)\), where \(\ualpha\in (0,\infty)^V\) and \(\beta > 0\), if it has the  density
		\begin{align*}
        f(\x)=K_G(\ualpha,\beta) \Delta_G(\x)^{\beta-1}\x^{\ualpha-1} \mathbbm{1}_{M_G}(\x),\qquad \x\in \mathbb{R}^V,
		\end{align*}
		where \(K_G(\ualpha,\beta)\) is a normalizing constant. This measure is denoted by \(\Dir_G(\ualpha,\beta)\).
		\item A probability measure on \(\mathbb{R}^V\) is the \(G\)-inverted Dirichlet distribution with parameters \((\ualpha,\beta)\), where \(\ualpha\in(0,\infty)^V\) and \(\beta > \max_{C\in \Ccal_G^+}|\ualpha_C|\), if it has the  density
		\begin{align*}
				f(\y)=k_G(\ualpha,\beta)\delta_G(\y)^{-\beta} \y^{\ualpha-1} \mathbbm{1}_{(0,\infty)^V}(\y),\qquad \y\in \mathbb{R}^V,
		\end{align*}
		where \(k_G(\ualpha,\beta)\) is a normalizing constant. This measure is denoted by \(\IDir_{G}(\ualpha,\beta)\).
	\end{enumerate}
\end{definition}
The normalizing constants are given by \cite[Theorem~5.5]{GDir}:
        \begin{align*}
				K_G(\ualpha,\beta)&=\frac{\prod_{C\in \Ccal_G^+}\Gamma(|\ualpha_C|+\beta)}{\Gamma(\beta)^m\prod_{i\in V}\Gamma(\alpha_i) \prod_{S\in \Scal_G^-}\Gamma(|\ualpha_S|+\beta)^{\nu_S}}\\
				k_G(\ualpha,\beta)&=\frac{\Gamma(\beta)^m \prod_{S\in \Scal_G^-} \Gamma(\beta-|\ualpha_S|)^{\nu_S}}{\prod_{i\in V}\Gamma(\alpha_i)\prod_{C\in \Ccal_G^+}\Gamma(\beta-|\ualpha_C|)},
		\end{align*}
where $m$ is the number of connected components of $G$. 

In \cite{GDir}, the graphical Dirichlet and graphical inverted Dirichlet
distributions were introduced as conjugate priors for the $\nm_G$ and $\mult_G$
models, respectively.  Placing these priors on the parameter vectors $\x$ or $\y$ and integrating them out leads to the marginal (or compound) distributions for the observed counts $\N$.

\begin{definition}[Compound graphical distributions]
Let $G$ be a decomposable graph.
\begin{enumerate}
    \item The $G$-Dirichlet-multinomial distribution,
    $\K \sim \DirMult_G(\ualpha,\beta,r)$, arises from the model
    $\K \mid \Y=\y \sim \mult_G(r,\y)$ with prior
    $\Y \sim \IDir_G(\ualpha,\beta)$. Its PMF is
    \begin{align*}
    \P(\K=\n)
    =
    \gbinom{r}{\n}{G}
    \frac{k_G(\ualpha,\beta)}{k_G(\ualpha+\n,\beta+r)}
    =
    \gbinom{r}{\n}{G}
    \E\left[
    \Y^\n \delta_G(\Y)^{-r}
    \right],
    \quad \n \in \mathbb{N}_{G,r}.
    \end{align*}

    \item The $G$-Dirichlet-negative multinomial distribution,
    $\K \sim \DirNm_G(\ualpha,\beta,r)$, arises from the model
    $\K \mid \X=\x \sim \nm_G(r,\x)$ with prior
    $\X \sim \Dir_G(\ualpha,\beta)$. Its PMF is
    \begin{align*}
    \P(\K=\n)
    =
    \gbinomm{|\n|+r-1}{\n}{G}
    \frac{K_G(\ualpha,\beta)}{K_G(\ualpha+\n,\beta+r)}
    =
    \gbinomm{|\n|+r-1}{\n}{G}
    \E\left[
    \X^\n \Delta_G(\X)^r
    \right],
    \quad \n \in \Natural^V.
    \end{align*}
\end{enumerate}
\end{definition}

These compound laws provide natural priors for the population
count vector \(\K\) in the graphical hypergeometric models. The Bayesian
hierarchical framework developed below follows this interpretation: a latent
graphical parameter \((\X\) or \(\Y)\) first generates \(\K\), and the observed
count vector \(\N\) is then obtained by sampling from this population without
replacement.

There are two parallel schemes.  In the fixed-draws case, the hierarchy is
\(
\Y\longrightarrow \K\longrightarrow \N,
\)
where
\[
\Y\sim \IDir_G(\ualpha,\beta),
\K\mid\Y \sim \mult_G(M,\Y),
\N\mid\K \sim \hg_G(M,\K,r),
\]
with the conditional independence relation $\N\perp\!\!\!\perp \Y\mid \K$.

In the fixed-failures case, the hierarchy is
\(
\X\longrightarrow \K\longrightarrow \N,
\)
where
\[
\X\sim \Dir_G(\ualpha,\beta),
\K\mid\X \sim \nm_G(M+1,\X),
\N\mid\K \sim \nhg_G(M,\K,r),
\]
with $\N\perp\!\!\!\perp \X\mid \K$.

Thus \(\K\) is interpreted as a random finite population count vector, while
\(\Y\) or \(\X\) is the latent graphical parameter governing its law.  The
same hierarchy also yields the induced marginal laws, posterior updates, and
predictive laws for the unobserved remainder.  
This Bayesian structure is summarized in
Table~\ref{tab:bayesian_conjugacy_schemes}.

\begin{table}[ht]
\centering
\caption{Bayesian hierarchies, posterior updates, and predictive structure.}
\label{tab:bayesian_conjugacy_schemes}
\scriptsize
\setlength{\tabcolsep}{3pt}
\renewcommand{\arraystretch}{1.55}
\begin{tabular}{@{}p{0.17\textwidth}p{0.27\textwidth}p{0.47\textwidth}@{}}
\toprule
\textbf{Stopping rule}
&
\textbf{Hierarchical model}
&
\textbf{Induced laws and updates}
\\
\midrule
 \begin{tabular}{c}   Fixed number\\of draws \end{tabular}
    &
    \(\begin{aligned}
        \Y &\sim \IDir_G(\ualpha,\beta),\\
        \K\mid\Y &\sim \mult_G(M,\Y),\\
        \N\mid\K &\sim \hg_G(M,\K,r),\\
        \N &\perp\!\!\!\perp \Y\mid\K
    \end{aligned}\)
    &
    \(\begin{aligned}
        \N \mid \Y &\sim \mult_G(r,\Y),\\
       \Y\mid \K &\sim \IDir_G(\ualpha+\K,\beta+M)\\
        \Y\mid \N &\sim \IDir_G(\ualpha+\N,\beta+r)\\
        \K &\sim \DirMult_G(\ualpha,\beta,M),\\
        \N &\sim \DirMult_G(\ualpha,\beta,r),\\
        \K-\N\mid\N
        &\sim
        \DirMult_G(\ualpha+\N,\beta+r,M-r),\\
                \K-\N\mid\Y, \N
        &\sim
        \mult_G(M-r,\Y)
    \end{aligned}\)
    \\

    \midrule
\begin{tabular}{c}
    Fixed number\\of failures \end{tabular}
    &
    \(\begin{aligned}
        \X &\sim \Dir_G(\ualpha,\beta),\\
        \K\mid\X &\sim \nm_G(M+1,\X),\\
        \N\mid\K &\sim \nhg_G(M,\K,r),\\
        \N &\perp\!\!\!\perp \X\mid\K
    \end{aligned}\)
    &
    \(\begin{aligned}
    \N\mid \X & \sim \nm_G(r,\X),\\
       \X\mid \K&\sim \Dir_G(\ualpha+\K, \beta+M+1),\\
        \X\mid \N&\sim \Dir_G(\ualpha+\N, \beta+r),\\
        \K &\sim \DirNm_G(\ualpha,\beta,M+1),\\
        \N &\sim \DirNm_G(\ualpha,\beta,r),\\
     \K-\N\mid\N
        &\sim
        \DirNm_G(\ualpha+\N,\beta+r,M-r+1),\\
        \K-\N\mid \X, \N
        &\sim
        \nm_G(M-r+1, \X)
    \end{aligned}\)
    \\

    \bottomrule
\end{tabular}
\end{table}

\subsection{Induced laws and posterior updates}

We first record the elementary splitting identities which explain the entries
in Table~\ref{tab:bayesian_conjugacy_schemes}.  They are the graphical
analogues of the classical binomial--hypergeometric and
negative-binomial--negative-hypergeometric identities.

\subsubsection{Fixed number of draws.}
Assume
\[
    \Y\sim\IDir_G(\ualpha,\beta),\quad
    \K\mid\Y=\y\sim\mult_G(M,\y),\quad
    \N\mid\K=\k\sim\hg_G(M,\k,r).
\]
Write
\(
    \L=\K-\N .
\)
Then, for admissible \(\n\in\mathbb N_{G,r}\) and
\(\ul\in\mathbb N_{G,M-r}\),
\[
\begin{aligned}
    \P(\N=\n,\L=\ul\mid \Y=\y)
    &=
    \P(\K=\n+\ul\mid\Y=\y)\,
    \P(\N=\n\mid \K=\n+\ul)
    \\
    &=
    \gbinom{M}{\n+\ul}{G}
    \frac{\y^{\n+\ul}}{\delta_G(\y)^M}
    \frac{
        \gbinom{r}{\n}{G}
        \gbinom{M-r}{\ul}{G}
    }{
        \gbinom{M}{\n+\ul}{G}
    }
    \\
    &=
    \left[
        \gbinom{r}{\n}{G}
        \frac{\y^\n}{\delta_G(\y)^r}
    \right]
    \left[
        \gbinom{M-r}{\ul}{G}
        \frac{\y^{\ul}}{\delta_G(\y)^{M-r}}
    \right].
\end{aligned}
\]
Consequently,
\(
    \N\mid\Y=\y\sim\mult_G(r,\y),
    \L\mid\Y=\y,\N=\n\sim\mult_G(M-r,\y),
\)
and in fact \(\N\perp\!\!\!\perp\L\mid\Y\).

Combining this with the inverted Dirichlet prior gives
\begin{align*}
   f_{\Y\mid \N=n}(\y)
    & \propto \P(\N=\n\mid \Y=\y) f_{\Y}(\y)
     = 
    \left[
        \gbinom{r}{\n}{G}
        \frac{\y^\n}{\delta_G(\y)^r}
    \right]
    k_G(\ualpha,\beta)
    \delta_G(\y)^{-\beta}
    \y^{\ualpha-\one}
\\ & \propto
    \delta_G(\y)^{-(\beta+r)}
    \y^{\ualpha+\n-\one}.
\end{align*}
Thus
\[
    \Y\mid\N=\n
    \sim
    \IDir_G(\ualpha+\n,\beta+r).
\]
Similarly, replacing \(\N=\n\) by \(\K=\k\) yields
\[
    \Y\mid\K=\k
    \sim
    \IDir_G(\ualpha+\k,\beta+M).
\]
Integrating out \(\Y\) gives
\[
    \K\sim\DirMult_G(\ualpha,\beta,M),
    \qquad
    \N\sim\DirMult_G(\ualpha,\beta,r),
\]
and
\[
    \K-\N\mid\N=\n
    \sim
    \DirMult_G(\ualpha+\n,\beta+r,M-r).
\]

\subsubsection{Fixed number of failures.}
Now assume
\[
    \X\sim\Dir_G(\ualpha,\beta),\quad
    \K\mid\X=\x\sim\nm_G(M+1,\x),\quad
    \N\mid\K=\k\sim\nhg_G(M,\k,r),
\]
and again write \(\L=\K-\N\).  For
\(\n,\ul\in\Natural^V\),
\begin{align*}
    \P(\N=\n,\L=\ul\mid \X=\x)
    &=
    \P(\K=\n+\ul\mid\X=\x)\,
    \P(\N=\n\mid \K=\n+\ul)
    \\
    &=
    \gbinomm{M+|\n+\ul|}{\n+\ul}{G}
    \x^{\n+\ul}\Delta_G(\x)^{M+1}
    \frac{
        \gbinomm{|\n|+r-1}{\n}{G}
        \gbinomm{M-r+|\ul|}{\ul}{G}
    }{
        \gbinomm{M+|\n+\ul|}{\n+\ul}{G}
    }
    \\
    &=
    \left[
        \gbinomm{|\n|+r-1}{\n}{G}
        \x^\n\Delta_G(\x)^r
    \right]
    \left[
        \gbinomm{M-r+|\ul|}{\ul}{G}
        \x^{\ul}\Delta_G(\x)^{M-r+1}
    \right].
\end{align*}

Therefore,
\(
    \N\mid\X=\x\sim\nm_G(r,\x),
    \L\mid\X=\x,\N=\n\sim\nm_G(M-r+1,\x),
\)
and \(\N\perp\!\!\!\perp\L\mid\X\).

The posterior update for the graphical Dirichlet parameter follows from
\begin{align*}
    f_{\X\mid \N=n}(\x)
    &\propto
    \left[
        \gbinomm{|\n|+r-1}{\n}{G}
        \x^\n\Delta_G(\x)^r
    \right]
    K_G(\ualpha,\beta)
    \Delta_G(\x)^{\beta-1}
    \x^{\ualpha-\one}\propto
    \Delta_G(\x)^{\beta+r-1}
    \x^{\ualpha+\n-\one}.
\end{align*}
Hence
\[
    \X\mid\N=\n
    \sim
    \Dir_G(\ualpha+\n,\beta+r),
\]
and, similarly,
\[
    \X\mid\K=\k
    \sim
    \Dir_G(\ualpha+\k,\beta+M+1).
\]
Integrating out \(\X\) gives
\begin{align*}
    \K\sim\DirNm_G(\ualpha,\beta,M+1),
    \N\sim\DirNm_G(\ualpha,\beta,r),
\end{align*}
and
\[
    \K-\N\mid\N=\n
    \sim
    \DirNm_G(\ualpha+\n,\beta+r,M-r+1).
\]

\section{Real data example}
\label{sec:application_rydberg}

As an empirical illustration of the graphical multinomial model, we consider
experimental data from the Rydberg atom study of Kim et al.~\cite{kim2022rydberg}.
In this experiment, each graph vertex is represented by an atom that is measured
in one of two relevant states: a non-excited state and a Rydberg excited state.
Thus, each measurement is naturally represented by a binary vector
\[
\N=(N_v)_{v\in V}\in\{0,1\}^{V},
\]
where \(N_v=1\) means that the atom corresponding to vertex \(v\) is observed in
the Rydberg excited state. The corresponding datasets for selected figures are
publicly available in the associated Figshare repository~\cite{kim2022data}.

We focus on four graph configurations appearing in Fig.~2 of
\cite{kim2022rydberg}: the path \(P_4\), the star \(S_4\), the cycle \(C_4\),
and the paw graph \(\mathrm{PAN}_4\); see Fig.~\ref{fig:rydberg_graphs}. Although the cycle \(C_4\) is not decomposable, the distribution \(\mult_{C_4}(1,\y)\) is well defined and globally Markov with respect to \(C_4\) by Remark \ref{rem:mult_one_arbitrary_G}. Hence all four graphs considered in this example are covered by the graphical Bernoulli model.

\begin{figure}[ht]
\centering
\begin{tikzpicture}[
    scale=0.9,
    every node/.style={circle, draw, inner sep=1.5pt, minimum size=6pt},
    every edge/.style={line width=0.7pt}
]

\node[draw=none, rectangle, anchor=east] at (-0.6,0.35) {$G$};
\node[draw=none, rectangle, anchor=east] at (-0.6,-2.35) {$G^\ast$};


\begin{scope}[xshift=0cm, yshift=0cm]
    \node (p1) at (0,0) {};
    \node (p2) at (0.9,0) {};
    \node (p3) at (0.9,0.9) {};
    \node (p4) at (0,0.9) {};
    \draw (p1)--(p2)--(p3)--(p4);
    \node[draw=none, rectangle, inner sep=0pt] at (0.9,-0.55) {$P_4$};
\end{scope}

\begin{scope}[xshift=3.8cm, yshift=0cm]
    \node (s1) at (0,0) {};
    \node (s2) at (-0.7,-0.5) {};
    \node (s3) at (0.7,-0.5) {};
    \node (s4) at (0,0.75) {};
    \draw (s1)--(s2);
    \draw (s1)--(s3);
    \draw (s1)--(s4);
    \node[draw=none, rectangle, inner sep=0pt] at (0,-0.95) {$S_4$};
\end{scope}

\begin{scope}[xshift=6.3cm, yshift=0cm]
    \node (c1) at (0,0) {};
    \node (c2) at (0.9,0) {};
    \node (c3) at (0.9,0.9) {};
    \node (c4) at (0,0.9) {};
    \draw (c1)--(c2)--(c3)--(c4)--(c1);
    \node[draw=none, rectangle, inner sep=0pt] at (0.45,-0.55) {$C_4$};
\end{scope}

\begin{scope}[xshift=8.9cm, yshift=0cm]
    \node (a1) at (0,0) {};
    \node (a2) at (0.9,0) {};
    \node (a3) at (0.45,0.8) {};
    \node (a4) at (1.45,0.8) {};
    \draw (a1)--(a2)--(a3)--(a1);
    \draw (a3)--(a4);
    \node[draw=none, rectangle, inner sep=0pt] at (0.7,-0.55)
    {$\mathrm{PAN}_4$};
\end{scope}


\begin{scope}[xshift=0cm, yshift=-2.7cm]
    \node (pp1) at (0,0) {};
    \node (pp2) at (0.9,0) {};
    \node (pp3) at (0.9,0.9) {};
    \node (pp4) at (0,0.9) {};
    \draw (pp1)--(pp3);
    \draw (pp1)--(pp4);
    \draw (pp2)--(pp4);
    \node[draw=none, rectangle, inner sep=0pt] at (0.9,-0.55)
    {$P_4^\ast$};
\end{scope}

\begin{scope}[xshift=3.8cm, yshift=-2.7cm]
    \node (ss1) at (0,0) {};
    \node (ss2) at (-0.7,-0.5) {};
    \node (ss3) at (0.7,-0.5) {};
    \node (ss4) at (0,0.75) {};
    \draw (ss2)--(ss3)--(ss4)--(ss2);
    \node[draw=none, rectangle, inner sep=0pt] at (0,-0.95)
    {$S_4^\ast$};
\end{scope}

\begin{scope}[xshift=6.3cm, yshift=-2.7cm]
    \node (cc1) at (0,0) {};
    \node (cc2) at (0.9,0) {};
    \node (cc3) at (0.9,0.9) {};
    \node (cc4) at (0,0.9) {};
    \draw (cc1)--(cc3);
    \draw (cc2)--(cc4);
    \node[draw=none, rectangle, inner sep=0pt] at (0.45,-0.55)
    {$C_4^\ast$};
\end{scope}

\begin{scope}[xshift=8.9cm, yshift=-2.7cm]
    \node (aa1) at (0,0) {};
    \node (aa2) at (0.9,0) {};
    \node (aa3) at (0.45,0.8) {};
    \node (aa4) at (1.45,0.8) {};
    \draw (aa1)--(aa4);
    \draw (aa2)--(aa4);
    \node[draw=none, rectangle, inner sep=0pt] at (0.7,-0.55)
    {$\mathrm{PAN}_4^\ast$};
\end{scope}

\end{tikzpicture}
\caption{Top row: the four graph configurations \(G\) used in the
support-restricted Rydberg-atom analysis. Bottom row: the corresponding
complement graphs \(G^\ast\). Independent sets of \(G\) are precisely cliques
of \(G^\ast\), and thus determine the support of the model.}
\label{fig:rydberg_graphs}
\end{figure}
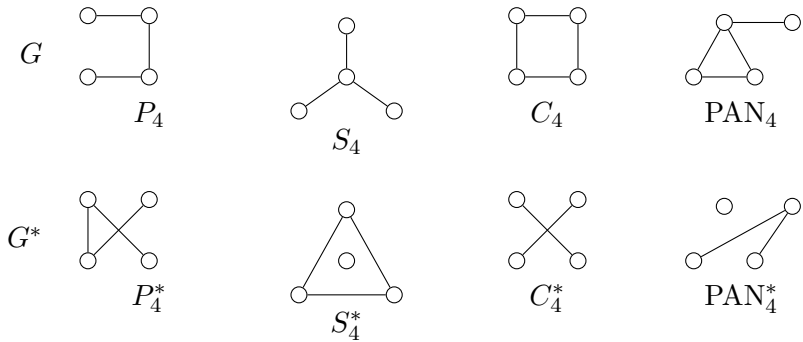

The relevance of independent sets comes from the Rydberg blockade mechanism.
When two atoms are sufficiently close, or are coupled through the engineered
Rydberg-wire construction, simultaneous excitation of both atoms is
energetically suppressed. In graph language, an edge
\(\{i,j\}\in E\) represents such an incompatibility: vertices \(i\) and \(j\)
should not both be observed in the Rydberg excited state in an ideal
blockade-consistent measurement. Therefore, for a given graph \(G=(V,E)\), the
natural support of the model is the family of independent sets of \(G\), or
equivalently, cliques of the complement graph \(G^\ast\). Under the model
\(\mult_G(1,\y)\), the probability mass function is
\[
\mathbb{P}(\N=\n)
=
\frac{\y^{\n}}{\delta_G(\y)},
\qquad
\n\in\mathbb{N}_{G,1}
=
\{\one_C\colon C\in\Ccal_{G^\ast}\}.
\]

The physical experiment is not perfectly described by this ideal support
constraint. State preparation and measurement errors, spontaneous decay, laser
phase noise, intensity fluctuations, finite temperature positional
fluctuations, and the projection steps used in the Rydberg-wire construction
can produce binary strings that violate the ideal blockade condition. Such
observations do not belong to the support of \(\mult_G(1,\y)\). For this
reason, the analysis is support-restricted: for each dataset, we retain only
the observations belonging to \(\mathbb{N}_{G,1}\) and fit
\(\mult_G(1,\y)\) to this admissible part of the sample. The analysis therefore
concerns blockade-consistent configurations rather than the complete
distribution of raw experimental shots.

More precisely, let
\[
\n^{(1)},\ldots,\n^{(m_{\text{adm}})}
\in\mathbb{N}_{G,1}
\]
denote the admissible observations. The log-likelihood is
\[
\ell(\y)
=
\sum_{k=1}^{m_{\text{adm}}}
\sum_{v\in V}
n_v^{(k)}\log y_v
-
m_{\text{adm}}\log\delta_G(\y),
\]
and the maximum likelihood estimator \(\widehat{\y}\) is obtained by
maximizing \(\ell(\y)\) over \(\y\in(0,\infty)^V\).
For numerical stability, the maximization was performed in the
log-parameters $\theta_v=\log y_v$ using the L-BFGS-B algorithm with
$\theta_v\in[-50,50]$. The same optimization procedure was used in each
bootstrap replication, and the imposed bounds were inactive in all fits.

The empirical frequencies of the admissible states are compared with the
fitted expected frequencies using the likelihood-ratio statistic
\[
T_{\mathrm{LR}}
=
2\sum_{\n\in\mathbb{N}_{G,1}}
O_{\n}
\log\left(\frac{O_{\n}}{E_{\n}}\right),
\]
where \(O_{\n}\) and \(E_{\n}\) denote the observed and fitted expected
frequencies of state \(\n\), respectively, and terms with \(O_{\n}=0\) are
set to zero.

Since \(\y\) is estimated from the data, the distribution of
\(T_{\mathrm{LR}}\) is approximated by a parametric bootstrap. In each
bootstrap replication, a sample of size \(m_{\text{adm}}\) is generated from
\(\mult_G(1,\widehat{\y})\), the parameter vector is re-estimated, and the
likelihood-ratio statistic is recomputed. For all analyses, we used \(B = 10000\) replications. The bootstrap \(p\)-value is
\[
p_{\mathrm{boot}}
=
\frac{
1+
\sum_{b=1}^{B}
\one_{\{
T_{\mathrm{LR}}^{(b)}
\geq
T_{\mathrm{LR}}^{\mathrm{obs}}
\}}
}{
B+1
}.
\]

The likelihood-ratio statistics and the corresponding bootstrap p-values are reported in Table \ref{tab:cut_only_results}. At the $5\%$ significance level, none of the fitted graphical multinomial models is rejected. 

\begin{table}[h!]
\centering
\caption{Support restricted goodness of fit results for the Rydberg atom
datasets. Here \(m\) denotes the number of raw experimental shots and
\(m_{\mathrm{adm}}\) the number of shots belonging to
\(\mathbb{N}_{G,1}\).}
\label{tab:cut_only_results}
\begin{tabular}{llrrrrr}
\toprule
Experiment & Graph & \(m\) & \(m_{\mathrm{adm}}\) &
\(m_{\mathrm{adm}}/m\) & \(T_{\mathrm{LR}}\) & \(p_{\mathrm{boot}}\) \\
\midrule
fig2d & \(P_4\)            & 678 & 191 & 28.17\% & 0.5081 & 0.9186 \\
fig2j & \(S_4\)            & 425 & 183 & 43.06\% & 8.9469 & 0.0863 \\
fig2e & \(C_4\)            & 672 & 261 & 38.84\% & 2.7639 & 0.2075 \\
fig2k & \(\mathrm{PAN}_4\) & 525 &  73 & 13.90\% & 1.5922 & 0.5326 \\
\bottomrule
\end{tabular}
\end{table}

The proportion of observations belonging to the theoretical support varies from
\(13.90\%\) for \(\mathrm{PAN}_4\) to \(43.06\%\) for \(S_4\). Thus, a substantial
fraction of the empirical excitation patterns violates the idealized graph-based
blockade constraint. Nevertheless, within the theoretical support, the goodness-of-fit
results indicate that the admissible excitation patterns are adequately described by
the fitted graphical multinomial distributions. Consequently, these distributions
should be interpreted as models for the conditional distribution of the excitation
pattern given that the observed configuration satisfies the corresponding graph-based
blockade constraint, rather than for the unconditional distribution of all observed
patterns.

The code used for data preprocessing, model fitting, and reproduction of
Table~\ref{tab:cut_only_results} is available in the accompanying
GitHub repository: \\\url{https://github.com/danielewskai/DiscreteParametricGraphicalModels}.

\section{Possible applications of clique and admissible-clique sampling}
\label{sec:independent_set_examples}

Many models with exclusion or incompatibility constraints can be represented
using independent sets of a graph \(G\), equivalently cliques of its complement
\(G^\ast\). Such structures arise, for example, in hard-core models from
statistical physics and in communication networks.

The four graphical distributions correspond to two sampling distinctions. For
sampling with replacement, \(\mult_G\) describes counts obtained from a fixed
number of clique draws, whereas \(\nm_G\) describes counts accumulated
until a prescribed number of failures. The distributions \(\hg_G\) and
\(\nhg_G\) are their respective without-replacement counterparts:
\(\hg_G\) corresponds to a fixed number of draws, while \(\nhg_G\) corresponds
to stopping after a fixed number of failures. In all four cases, the elementary
feasible configurations are cliques \(C\in\Ccal_{G^\ast}\), represented by
their indicator vectors \(\zeta=\one_C\). For the negative models, these
configurations form admissible sequences of cliques.

In applications, the graph-based support should be viewed as an idealization.
Configurations outside the support may occur because of measurement errors,
imperfect exclusion mechanisms, omitted interactions, or misspecification of
the graph \(G\). The models may therefore be used directly for admissible
observations, conditionally on admissibility, or as idealized null models.

\subsection{Wireless networks and CSMA scheduling}

Wireless scheduling is a natural source of clique indicators. The vertices of
\(G\) may represent wireless links, and an edge \(\{u,v\}\in E\) indicates that
the corresponding links interfere and cannot be active in the same time slot.
A collision-free schedule is therefore an independent set of \(G\),
equivalently a clique of \(G^\ast\).

If \(C^{(t)}\in\Ccal_{G^\ast}\) is the feasible schedule used in time slot
\(t\), then \(\zeta^{(t)}=\one_{C^{(t)}}\) records the links active in that slot,
and
\[
N_v=\sum_{t=1}^r \zeta_v^{(t)}
\]
counts the number of sampled slots in which link \(v\) is active. In idealized
CSMA models with fixed link fugacities, the schedule process has a stationary
distribution of the product form $\P(C)\propto 
\prod_{v\in C}y_v$ for $C\in\Ccal_{G^\ast}$
for suitable \(\y\). This is precisely the clique distribution underlying
\(\mult_G(1,\y)\). Hence, for a decomposable graph \(G\), \(r\) independent
stationary schedule samples give $\N\sim\mult_G(r,\y)$.
Schedules in consecutive slots are generally dependent, so this exact
interpretation applies to independent stationary samples; for observations
taken sufficiently far apart, it may also provide an approximate model.

The distribution \(\hg_G\) provides the corresponding finite-content,
without-replacement model. More precisely, if an ordered collection of
\(M\) feasible schedules is chosen uniformly among all such collections
having prescribed aggregate activity vector \(\K\), then the aggregate
activity vector in \(r\) of the \(M\) positions has distribution
\(\hg_G(M,\K,r)\).

This conflict-graph construction is standard in the CSMA literature:
feasible schedules are independent sets of the conflict graph, and
product-form stationary distributions of the form above arise in idealized
CSMA models
\cite{boorstyn1987throughput,jiang2010csma,
rajagopalan2009network_adiabatic}.

The failure-stopped construction of \(\nm_G\) also suggests an idealized
model for communication episodes. Here \(C_t\) represents the set of links
successfully active at stage \(t\). In the exact construction described in
Section~\ref{sec:nmg_int}, successive nonempty feasible schedules satisfy the
Cartier--Foata admissibility relation \(C_t\to C_{t+1}\), with the specified
transition probabilities, and the episode terminates upon reaching the empty
clique. Under this particular mechanism, the accumulated activity vector
\[
\N=\sum_{t=0}^{T-1}\one_{C_t}
\]
has distribution \(\nm_G(1,\x)\). Admissibility alone does not determine this
law, so for general CSMA dynamics \(\nm_G\) should be regarded as an
idealized model rather than an exact consequence of the interference
constraints. The distribution \(\nhg_G\) provides the corresponding
finite-content, without-replacement model under the sampling construction
described in Section~\ref{sec:clique-nhg}.

\subsection{Cancer genomics and mutually exclusive alterations}

In cancer genomics, the vertices of \(G\) may represent genes, pathways, or
genomic alteration events. An edge \(\{u,v\}\in E\) indicates that alterations
\(u\) and \(v\) are regarded as mutually exclusive, in the sense that they are
not expected to occur together in the same tumour sample. Under a
hard-constraint idealization, each tumour profile is therefore an independent
set of \(G\), or equivalently a clique of \(G^\ast\).

For tumour sample \(t\), let
\[
C^{(t)}
=
\{v\in V\colon \text{alteration }v\text{ is present in sample }t\},
\qquad
\zeta^{(t)}=\one_{C^{(t)}}.
\]
Then $N_v=\sum_t \zeta_v^{(t)}$
is the number of samples in which alteration \(v\) occurs. If tumour profiles
are treated as approximately independent draws from a common distribution on
the admissible support, then \(\mult_G(r,\y)\) gives the corresponding
with-replacement model. The distribution \(\hg_G\) gives the
without-replacement counterpart, for example when a subset is sampled from a
finite cohort while aggregate alteration counts in the full cohort are fixed.
This is analogous to exact procedures that condition on the aggregate
alteration counts in the full cohort.

This application is related to the extensive literature on mutually exclusive
cancer modules, including Dendrix and Multi-Dendrix
\cite{vandin2012dendrix,leiserson2013multidendrix}, MEMo
\cite{ciriello2012memo}, CoMEt \cite{leiserson2015comet}, network-based
approaches \cite{babur2015mutual}, and weighted or permutation-based
significance tests such as WExT and WeSME
\cite{leiserson2016wext,kim2016wesme}. Large public cancer-genomics resources,
such as cBioPortal, provide data in which such models may be explored
\cite{cerami2012cbioportal,gao2013cbioportal}.

The support restriction is particularly delicate in this application because
mutual exclusivity is usually a statistical tendency rather than a deterministic
rule. A profile outside the independent-set support may reflect tumour
heterogeneity, subtype structure, pathway cross-talk, measurement error, or an
incomplete incompatibility graph. Consequently, the fitted graphical model is
most naturally interpreted as an idealized null model or as the conditional law
of a tumour profile given that it satisfies the proposed exclusion constraints.

\subsection{Spatial hard-core configurations, adsorption, and packing}

A broad class of examples comes from spatial statistics and statistical physics.
The vertices of \(G\) represent candidate locations or objects, such as particles,
trees, cells, sensors, transmitters, or adsorbed molecules.  An edge
\(\{u,v\}\in E\) means that the two candidates are too close, overlap, or
otherwise cannot be simultaneously present.  A feasible configuration is then an
independent set of \(G\), or a clique of \(G^\ast\).

After discretizing a spatial observation window or constructing a finite set of
candidate objects, one observation gives $\zeta=\one_C$, $C\in\Ccal_{G^\ast}$.
Repeated spatial samples give counts
$\N=\sum_{j=1}^r \one_{C^{(j)}}$,
which may be modelled by \(\mult_G(r,\y)\) under an i.i.d. sampling approximation.
The model \(\hg_G\) is again the finite-population or conditioned analogue.

This interpretation is connected to Matérn hard-core processes, Strauss-type
repulsive point processes, and other spatial hard-core models
\cite{matern_spatial_variation,stoyan_spatial_point_patterns}.  It is
also related to the hard-core lattice gas and hard-square or hard-hexagon models,
where neighbouring occupied sites are forbidden and the partition function is an
independence polynomial \cite{baxter_hard_hexagons,Scott_2005}.  In the notation
of this paper, the single-configuration law is described as
\(\mult_G(1,\y)\), because it is a distribution on indicators of cliques of
\(G^\ast\).

Random sequential adsorption and packing models provide a natural bridge to the
failure-stopped mechanism.  In such models, objects are proposed sequentially and
accepted only if they are compatible with the current configuration.  A stage
\(t\) may be represented by a clique \(C_t\) of mutually compatible accepted
objects, and the process stops when a proposed object is rejected, when no
further object can be added, or when a jamming criterion is reached.  This
suggests \(\nm_G\)-type models for accumulated counts until failure, with
\(\nhg_G\) as the corresponding conditioned finite-population version.  Real
adsorption processes can be strongly path dependent, so the graphical negative
models should be viewed as idealized stopped-sampling models rather than as
literal descriptions of all adsorption dynamics \cite{cadilhe_rsa}.

\subsection{Loss networks, resource sharing, and abstract polymer models}

Further examples arise in stochastic networks and statistical mechanics. In a
loss network or resource-sharing system with pairwise incompatibility
constraints, the vertices of \(G\) may represent calls, routes, jobs, or
resource requests. An edge \(\{u,v\}\in E\) indicates that the two requests
cannot be simultaneously active because they compete for limited capacity.
Thus the active set at a fixed time is an independent set of \(G\). Repeated
stationary snapshots give clique indicators
\(\zeta^{(t)}=\one_{C^{(t)}}\) and occupancy counts
\[
\N=\sum_{t=1}^r \zeta^{(t)}.
\]
When these snapshots are treated as independent draws, \(\mult_G\) gives the
corresponding with-replacement model, whereas \(\hg_G\) gives its
without-replacement counterpart for sampling from a finite collection of
snapshots with fixed aggregate occupancy counts. Product-form stationary
distributions in loss networks and related stochastic network models provide a
classical motivation for such graph-constrained snapshot laws
\cite{kelly_loss_networks}.

In abstract polymer models, the vertices of \(G\) represent polymers and an edge
indicates incompatibility. An admissible polymer configuration is a set of
mutually compatible polymers, hence a clique of \(G^\ast\). The polynomial
\(\delta_G\) is then the finite-volume partition function in the notation of
this paper. A single random polymer configuration is described by
\(\mult_G(1,\y)\), while repeated with-replacement sampling gives
\(\mult_G(r,\y)\). The corresponding without-replacement model is
\(\hg_G\). This connects the graphical models to the large literature on
polymer systems, cluster expansions, and algorithms for hard-core partition
functions
\cite{kotecky_preiss,fernandez_procacci,weitz_counting_independent_sets}.

A stopped admissible-clique mechanism may be considered when a network or
polymer system is observed over an episode rather than at a single snapshot.
The clique \(C_t\) records the compatible set of activities, jobs, or polymers
present in layer \(t\), and the episode stops at an idle, blocking, rejection,
or failure state. The resulting accumulated count vector has the qualitative
form described by the with-replacement model \(\nm_G\), while \(\nhg_G\)
provides the corresponding without-replacement model with prescribed aggregate
content.

\subsection{Summary of the modelling role}

The examples above identify settings in which observations are naturally
represented by cliques of \(G^\ast\), or by admissible sequences of such
cliques. The distributions \(\mult_G\) and \(\nm_G\) describe the corresponding
with-replacement sampling schemes, with stopping after a fixed number of draws
or failures, respectively. The distributions \(\hg_G\) and \(\nhg_G\) are
their without-replacement counterparts. In each application, the principal
modelling choice is the specification of the incompatibility graph \(G\), whose
adequacy should be assessed against the observed support.

\appendix

\section{Proofs}
\label{app:proofs}

\subsection{Proofs from Section \ref{sec:hypergeometric_models} }

\begin{proof}[Proof of Proposition \ref{prop:hg}]
For any $\n\in\mathbb{N}^V$ such that $0\le n_v \le K_v, v\in V$,
\begin{align*}
\P(\N_1=\n \mid \N_1+\N_2=\K)
  &= \frac{\P(\N_1=\n)\,\P(\N_2=\K-\n)}{\P(\N_1+\N_2=\K)} \\
  &= \frac{
       \gbinom{r}{\n}{G}(\delta_G(\y))^{-r}\prod_{v\in V}y_v^{n_v}\,
       \gbinom{M-r}{\K-\n}{G}(\delta_G(\y))^{-(M-r)}\prod_{v\in V}y_v^{K_v-n_v}
     }{
       \gbinom{M}{\K}{G}(\delta_G(\y))^{-M}\prod_{v\in V}y_v^{K_v}
     } \\
  &= \frac{\gbinom{r}{\n}{G}\,\gbinom{M-r}{\K-\n}{G}}
          {\gbinom{M}{\K}{G}}.
\end{align*}
This is the PMF of the $G$-hypergeometric distribution with parameters $(M,\K,r)$.
\end{proof}

\begin{proof}[Proof of Theorem~\ref{thm:properties_hg}]
We prove the two assertions separately.

\noindent
(1) Global Markov property.

Applying Lemma~\ref{lem:coeff_factorization} to the three
graph-multinomial coefficients in Definition~\ref{def:hg_G} gives
\[
\P(\N=\n)
=
\frac{
\displaystyle\prod_{C\in\Ccal_G^+}
\hg(M,\K_C,r)(\n_C)
}{
\displaystyle\prod_{S\in\Scal_G^-}
\hg(M,\K_S,r)(\n_S)^{\nu_S}
}.
\]
The displayed clique-separator factorization implies the global Markov
property by the observation following \eqref{eq:clique-separator-factorization}.

\noindent
(2) Clique marginals.

For every \(S\subseteq C\), the \(S\)-marginal of
\(\hg(M,\K_C,r)\) is \(\hg(M,\K_S,r)\). Consequently, the clique
and separator distributions appearing in the factorization form a
consistent family. By the standard clique-separator construction for
decomposable graphs, the factorization above defines a probability
distribution whose maximal-clique marginals satisfy $\N_C\sim\hg(M,\K_C,r)$ for all $C\in\Ccal_G^+$.

For an arbitrary clique \(C\), choose a maximal clique \(C^+\supseteq C\).
Since \(\N_C\) is a marginal of \(\N_{C^+}\), and marginals of a multivariate
hypergeometric distribution are again hypergeometric, $\N_C\sim\hg(M,\K_C,r)$.
\end{proof}

\begin{proof}[Proof of Theorem~\ref{thm:dag-hg}] Applying the DAG factorization of the graph-multinomial coefficients of the first type from \cite[Lemma~2.15(2)]{GDir} to the three coefficients in the PMF of \(\hg_G(M,\K,r)\), we obtain \[ \P(\N=\n) = \prod_{v\in V} \frac{ \binom{r-|\n_{\pa(v)}|}{n_v} \binom{ M-r-|\K_{\pa(v)}|+|\n_{\pa(v)}| }{ K_v-n_v } }{ \binom{M-|\K_{\pa(v)}|}{K_v} }. \] For each \(v\in V\), the corresponding factor is the PMF of the classical hypergeometric distribution \[ \hg\left( M-|\K_{\pa(v)}|, K_v, r-|\n_{\pa(v)}| \right) \] evaluated at \(n_v\). Therefore, \[ \P(\N=\n) = \prod_{v\in V} \hg\left( M-|\K_{\pa(v)}|, K_v, r-|\n_{\pa(v)}| \right)(n_v). \] This is the DAG factorization with respect to \(\Gcal\), and hence \[ \P\!\left( N_v=n_v \mid \N_{\pa(v)}=\n_{\pa(v)} \right) = \hg\left( M-|\K_{\pa(v)}|, K_v, r-|\n_{\pa(v)}| \right)(n_v), \qquad v\in V. \] \end{proof}

\begin{proof}[Proof of Proposition \ref{prop:nhg}]
Let $\N_1 \sim \nm_G(r,\x)$ and $\N_2 \sim \nm_G(M-r+1,\x)$ be independent.  
For any $\n \leq \K$ we have
\begin{align*}
\P(\N_1=\n &\mid \N_1+\N_2=\K)
  = \frac{\P(\N_1=\n)\,\P(\N_2=\K-\n)}{\P(\N_1+\N_2=\K)} \\
  &= \frac{
       \gbinomm{r+|\n|-1}{\n}{G}\,\delta_G^r(-\x)\prod_{v\in V}x_v^{n_v}\,
       \gbinomm{M-r+|\K|-|\n|}{\K-\n}{G}\,\delta_G^{M-r+1}(-\x)\prod_{v\in V}x_v^{K_v-n_v}
     }{
       \gbinomm{M+|\K|}{\K}{G}\,\delta_G^{M+1}(-\x)\prod_{v\in V}x_v^{K_v}
     } \\
  &= \frac{
       \gbinomm{r+|\n|-1}{\n}{G}\,
       \gbinomm{M-r+|\K|-|\n|}{\K-\n}{G}
     }{
       \gbinomm{M+|\K|}{\K}{G}
     }.
\end{align*}
This is exactly the PMF of the $G$-negative hypergeometric distribution with parameters $(M,\K,r)$. 
\end{proof}

\begin{proof}[Proof of Theorem~\ref{thm:properties_nhg}]
We prove the two assertions separately.

\noindent
(1) Global Markov property.

Applying Lemma~\ref{lem:coeff_factorization} to the three graphical
coefficients in the definition of \(\nhg_G(M,\K,r)\) gives
\[
\P(\N=\n)
=
\frac{
\displaystyle\prod_{C\in\Ccal_G^+}
\nhg(M,\K_C,r)(\n_C)
}{
\displaystyle\prod_{S\in\Scal_G^-}
\nhg(M,\K_S,r)(\n_S)^{\nu_S}
}.
\]
The displayed clique--separator factorization implies the global Markov
property by the observation following \eqref{eq:clique-separator-factorization}.

\noindent
(2) Clique marginals.

For every \(S\subseteq C\), the \(S\)-marginal of
\(\nhg(M,\K_C,r)\) is \(\nhg(M,\K_S,r)\). By the standard clique-separator construction for
decomposable graphs, the factorization above defines a probability
distribution whose maximal-clique marginals satisfy $\N_C\sim\nhg(M,\K_C,r)$ for all $C\in\Ccal_G^+$.

For an arbitrary clique \(C\), choose a maximal clique \(C^+\supseteq C\).
Since \(\N_C\) is a marginal of \(\N_{C^+}\), and marginals of a multivariate
negative hypergeometric distribution are again negative hypergeometric, $\N_C\sim\nhg(M,\K_C,r)$.
\end{proof}

\begin{proof}[Proof of Theorem~\ref{thm:dag-nhg}]
Applying the DAG factorization of the graph-multinomial coefficients of the
second type from \cite[Lemma~2.15(2)]{GDir} to the three coefficients in the
PMF of \(\nhg_G(M,\K,r)\), we obtain
\[
\P(\N=\n)
=
\prod_{v\in V}
\frac{
\binom{n_v+|\n_{\pa(v)}|+r-1}{n_v}
\binom{
M-r+|\K_{\pa(v)}|-|\n_{\pa(v)}|+K_v-n_v
}{
K_v-n_v
}
}{
\binom{M+|\K_{\pa(v)}|+K_v}{K_v}
}.
\]
For each \(v\in V\), the corresponding factor is the PMF of the classical
negative hypergeometric distribution
\[
\nhg\left(
M+|\K_{\pa(v)}|,
K_v,
r+|\n_{\pa(v)}|
\right)
\]
evaluated at \(n_v\). Therefore,
\[
\P(\N=\n)
=
\prod_{v\in V}
\nhg\left(
M+|\K_{\pa(v)}|,
K_v,
r+|\n_{\pa(v)}|
\right)(n_v).
\]
This is the DAG factorization with respect to \(\Gcal\), and hence
\[
\P\!\left(
N_v=n_v
\mid
\N_{\pa(v)}=\n_{\pa(v)}
\right)
=
\nhg\left(
M+|\K_{\pa(v)}|,
K_v,
r+|\n_{\pa(v)}|
\right)(n_v),
\qquad v\in V.
\]
\end{proof}

\subsection{Proofs from Section \ref{sec:interpretation}}

\begin{proof}[Proof of Theorem~\ref{thm:interp_hg}]
We first record the following combinatorial interpretation of the graph-multinomial coefficients. Since
\[
\delta_G(\y)
=
\sum_{C\in\Ccal_{G^*}}\y^{\one_C},
\]
we have
\[
\sum_{\n \in \Natural_{G,r}} \gbinom{r}{\n}{G} \y^{\n} = \delta_G(\y)^r
=
\sum_{(C^{(1)},\dots,C^{(r)})\in\Ccal_{G^*}^r}
\y^{\one_{C^{(1)}}+\cdots+\one_{C^{(r)}}}.
\]
Therefore, by comparing the coefficients of $\y^\n$, \(\gbinom{r}{\n}{G}\) is the number of ordered \(r\)-tuples $(C^{(1)},\dots,C^{(r)})\in\Ccal_{G^*}^r$
such that
\[
\sum_{j=1}^r \one_{C^{(j)}}=\n.
\]

Now fix \(\n\in\Natural^V\). The event \(\{\N=\n\}\) occurs exactly when the first \(r\) cliques have total content \(\n\), while the remaining \(M-r\) cliques have total content \(\K-\n\). Thus,
\[
\left| \{\N=\n\}\right| = \gbinom{r}{\n}{G} \gbinom{M-r}{\K-\n}{G}.
\]
Since the total number of ordered \(M\)-tuples in \(\mathcal B_M(\K)\) is $\gbinom{M}{\K}{G}$,
uniform sampling from \(\mathcal B_M(\K)\) gives
\[
\mathbb P(\N=\n)
=
\frac{
\gbinom{r}{\n}{G}
\gbinom{M-r}{\K-\n}{G}
}{
\gbinom{M}{\K}{G}
}.
\]
This is exactly the probability mass function of \(\hg_G(M,\K,r)\).
\end{proof}

\begin{proof}[Proof of Lemma~\ref{lem:propP}]
We prove the following identity: for every \(S\subseteq V\),
\begin{equation}
\label{eq:restricted-p-identity}
    \sum_{\substack{C\in\Ccal_{G^*}\\ C\subseteq S}}
    p_{\x}(C)
    =
    \Delta_{G_{V\setminus S}}(\x_{V\setminus S}).
\end{equation}
Indeed, expanding the left-hand side gives
\[
\begin{aligned}
\sum_{\substack{C\in\Ccal_{G^*}\\ C\subseteq S}}
p_{\x}(C)
&=
\sum_{\substack{C\in\Ccal_{G^*}\\ C\subseteq S}}
\x^{\one_C}
\sum_{B\in\Ccal_{(G_{V\setminus \onb_G(C)})^*}}
(-1)^{|B|}\x^{\one_B}.
\end{aligned}
\]
For each pair \((C,B)\) appearing in the sum, the union \(A=C\cup B\) is a
clique of \(G^*\). Conversely, for a fixed \(A\in\Ccal_{G^*}\), the possible
choices of \(C\) are precisely the subsets \(C\subseteq A\cap S\), with
\(B=A\setminus C\). Therefore the coefficient of \(\x^{\one_A}\) in the
last display is
\[
    \sum_{C\subseteq A\cap S}(-1)^{|A|-|C|}.
\]
This coefficient is equal to \((-1)^{|A|}\) if \(A\cap S=\emptyset\), and is
zero otherwise. Hence only cliques \(A\subseteq V\setminus S\) remain, and so
\[
    \sum_{\substack{C\in\Ccal_{G^*}\\ C\subseteq S}}
    p_{\x}(C)
    =
    \sum_{A\in\Ccal_{(G_{V\setminus S})^*}}
    (-1)^{|A|}\x^{\one_A}
    =
    \Delta_{G_{V\setminus S}}(\x_{V\setminus S}).
\]
This proves \eqref{eq:restricted-p-identity}.

Taking \(S=V\) gives
\[
    \sum_{C\in\Ccal_{G^*}}p_{\x}(C)=\Delta_{G_\emptyset}=1.
\]
Moreover, since \(\x\in M_G\),
\[
    \Phi_{\x}(C)
    =
    \Delta_{G_{V\setminus \onb_G(C)}}(\x_{V\setminus \onb_G(C)})>0
\]
for every \(C\in\Ccal_{G^*}\). Since \(\x^{\one_C}>0\), we get
\(p_{\x}(C)>0\).

Finally, let \(C\in\Acal_{G^*}\). Applying
\eqref{eq:restricted-p-identity} with \(S=\onb_G(C)\), we obtain
\[
    \Phi_{\x}(C)
    =
    \sum_{\substack{D\in\Ccal_{G^*}\\ D\subseteq \onb_G(C)}}
    p_{\x}(D).
\]
Since \(D\subseteq\onb_G(C)\) is equivalent to \(D=\emptyset\) or \(C\to D\),
this gives
\[
    \Phi_{\x}(C)
    =
    p_{\x}(\emptyset)
    +
    \sum_{D\colon C\to D}p_{\x}(D).
\]
\end{proof}

\begin{proof}[Proof of Theorem~\ref{prop:clique-chain-nm1}]
Fix a nonempty admissible sequence
\[
    (C_0,\dots,C_{m-1}),
    \qquad
    C_t\in\Acal_{G^*},\quad C_t\to C_{t+1},
    \quad t=0,\dots,m-2.
\]
The probability that the procedure outputs exactly this sequence is
\[
\mathbb P(C_0,\dots,C_{m-1},C_m=\emptyset)  =
p_{\x}(C_0)
\prod_{t=0}^{m-2}
\frac{p_{\x}(C_{t+1})}{\Phi_{\x}(C_t)}
\cdot
\frac{p_{\x}(\emptyset)}{\Phi_{\x}(C_{m-1})}.
\]
Using the definition of $p_{\x}$ and telescoping, 
the right-hand side can be written as 
\[
\x^{\one_{C_0}}\Phi_{\x}(C_0)
\prod_{t=0}^{m-2}
\frac{
    \x^{\one_{C_{t+1}}}\Phi_{\x}(C_{t+1})
}{
    \Phi_{\x}(C_t)
}
\cdot
\frac{\Delta_G(\x)}{\Phi_{\x}(C_{m-1})} 
=
\Delta_G(\x)\,
\x^{\one_{C_0}+\cdots+\one_{C_{m-1}}}.
\]
For the empty sequence, the probability is
\[
    \mathbb P(C_0=\emptyset)
    =
    p_{\x}(\emptyset)
    =
    \Delta_G(\x),
\]
which is the same formula with the empty sum in the exponent.

By Lemma \ref{lem:CF-normal-form}, finite admissible sequences of nonempty cliques of \(G^*\) are in one-to-one correspondence with elements of the free quotient monoid \(L\). Moreover, if the admissible sequence
\((C_0,\dots,C_{m-1})\) corresponds to \(\ell\in L\), then
\[
    \varepsilon(\ell)
    =
    \one_{C_0}+\cdots+\one_{C_{m-1}}.
\]
Hence the procedure assigns to each \(\ell\in L\) probability $   \Delta_G(\x)\x^{\varepsilon(\ell)}$. 
In particular, the procedure terminates almost surely, since
\[
    \sum_{\ell\in L}\Delta_G(\x)\x^{\varepsilon(\ell)}
    =
    \Delta_G(\x)
    \sum_{\ell\in L}\x^{\varepsilon(\ell)}
    =
    \Delta_G(\x)\Delta_G(\x)^{-1}
    =
    1,
\]
where the Cartier--Foata identity was used.

Therefore, for \(\n\in\Natural^V\),
\begin{align*}
    \mathbb P(\N=\n)=
    \sum_{\ell\in L\colon \varepsilon(\ell)=\n}
    \Delta_G(\x)\x^\n  =
    \left|\{\ell\in L\colon \varepsilon(\ell)=\n\}\right|
    \Delta_G(\x)\x^\n  =
    \gbinomm{|\n|}{\n}{G}
    \x^\n
    \Delta_G(\x)
\end{align*}
where in the last equality we used \cite[(40)]{GDir}. This is exactly the probability mass function of \(\nm_G(1,\x)\).
\end{proof}

\begin{proof}[Proof of Theorem~\ref{thm:interp_nhg}]
We first recall the counting interpretation of the coefficients appearing in
\(\Delta_G(\x)^{-q}\). By the Cartier--Foata normal form, admissible sequences
of nonempty cliques of \(G^*\), together with the empty sequence, are in
one-to-one correspondence with elements of the free quotient monoid \(L\) induced by $G$. Hence, for
\(q\in\Natural_+\),
\[
    \Delta_G(\x)^{-q}
    =
    \left(
    \sum_{\mathbf C}
    \x^{\varepsilon(\mathbf C)}
    \right)^q
    =
    \sum_{\n\in\Natural^V}
    |\mathcal A_q(\n)|\,\x^\n,
\]
where the sum  in the middle expression is taken over all admissible sequences.
On the other hand, by definition of the graph-negative multinomial coefficients,
\[
    \Delta_G(\x)^{-q}
    =
    \sum_{\n\in\Natural^V}
    \gbinomm{|\n|+q-1}{\n}{G}\x^\n.
\]
Comparing coefficients gives
\[
    |\mathcal A_q(\n)|
    =
    \gbinomm{|\n|+q-1}{\n}{G},
    \qquad q\in\Natural_+,\ \n\in\Natural^V.
\]

Now fix \(\n\in\Natural^V\). The event \(\{\N=\n\}\) occurs precisely when the
first \(r\) admissible sequences have total content \(\n\), while the remaining
\(M+1-r\) admissible sequences have total content \(\K-\n\). Therefore the
number of elements of \(\mathcal A_{M+1}(\K)\) giving rise to \(\N=\n\) is
\[
    |\mathcal A_r(\n)|\,
    \cdot |\mathcal A_{M+1-r}(\K-\n)|=
    \gbinomm{|\n|+r-1}{\n}{G}\cdot
    \gbinomm{|\K-\n|+M-r}{\K-\n}{G}.
\]

Since the total number of admissible \((M+1)\)-tuples with content \(\K\) is
\[
    |\mathcal A_{M+1}(\K)|
    =
    \gbinomm{|\K|+M}{\K}{G},
\]
uniform sampling from \(\mathcal A_{M+1}(\K)\) gives
\[
    \mathbb P(\N=\n)
    =
    \frac{
        \gbinomm{|\n|+r-1}{\n}{G}
        \gbinomm{|\K-\n|+M-r}{\K-\n}{G}
    }{
        \gbinomm{|\K|+M}{\K}{G}
    }.
\]
This is exactly the probability mass function of \(\nhg_G(M,\K,r)\). 
\end{proof}

\section*{Declaration of generative AI and AI-assisted technologies in the writing process}
During the preparation of this work the authors used ChatGPT (versions 5.5 and 5.6) in order to assist with literature searches, verify mathematical derivations, and write the code used for the data analysis presented in Section \ref{sec:application_rydberg}. After using this tool, the authors reviewed and edited the content as needed and take full responsibility for the content of the publication.

\bibliographystyle{plain}
\bibliography{Bibl}

@incollection {DiekertMetivier1997,
    AUTHOR     = {Diekert, V. and M\'{e}tivier, Y.},
    TITLE      = {Partial commutation and traces},
    BOOKTITLE  = {Handbook of formal languages, {V}ol. 3},
    PAGES      = {457--533},
    PUBLISHER  = {Springer, Berlin},
    YEAR       = {1997},
}

@article {Mult24,
    AUTHOR   = {Kus, D. and Singh, K. and Venkatesh, R.},
    TITLE    = {Identities of the multi-variate independence polynomials from heaps theory},
    JOURNAL  = {Proc. Indian Acad. Sci. Math. Sci.},
    FJOURNAL = {Proceedings of the Indian Academy of Sciences. Mathematical Sciences},
    VOLUME   = {134},
    YEAR     = {2024},
    NUMBER   = {1},
    PAGES    = {Paper No. 16, 11},
}

@article {Scott_2005,
    AUTHOR     = {Scott, A. D. and Sokal, A. D.},
    TITLE      = {The repulsive lattice gas, the independent-set polynomial, and the {L}ov\'{a}sz local lemma},
    JOURNAL    = {J. Stat. Phys.},
    FJOURNAL   = {Journal of Statistical Physics},
    VOLUME     = {118},
    YEAR       = {2005},
    NUMBER     = {5-6},
    PAGES      = {1151--1261},
}

@misc {GDir,
    AUTHOR = {Danielewska, I. and Ko{\l}odziejek, B. and Weso{\l}owski, J. and Zeng, X.},
    TITLE  = {Graphical negative multinomial and multinomial models with {Dirichlet}-type priors},
    YEAR   = {2025},
    NOTE   = {Preprint, arXiv:2301.06058v5 [math.PR]},
}

@book {cartier1969commutation,
    AUTHOR     = {Cartier, P. and Foata, D.},
    TITLE      = {Probl\`emes combinatoires de commutation et r\'{e}arrangements},
    SERIES     = {Lecture Notes in Mathematics, No. 85},
    PUBLISHER  = {Springer-Verlag, Berlin-New York},
    YEAR       = {1969},
    PAGES      = {iv+88},
}

@inproceedings {levit2005independence,
    AUTHOR    = {Levit, V. E. and Mandrescu, E.},
    TITLE     = {The independence polynomial of a graph---a survey},
    BOOKTITLE = {Proceedings of the 1st International Conference on Algebraic Informatics},
    PAGES     = {233--254},
    PUBLISHER = {Aristotle University of Thessaloniki},
    ADDRESS   = {Thessaloniki},
    YEAR      = {2005},
}

@techreport {goldwurm1998clique,
    AUTHOR      = {Goldwurm, M. and Saporiti, L.},
    TITLE       = {Clique polynomials and trace monoids},
    INSTITUTION = {Dipartimento di Scienze dell'Informazione, Universit{\`a} degli Studi di Milano},
    TYPE        = {Rapporto Interno},
    NUMBER      = {222-98},
    ADDRESS     = {Milano, Italy},
    YEAR        = {1998},
}

@article {Dawid1993,
    AUTHOR   = {Dawid, A. P. and Lauritzen, S. L.},
    TITLE    = {Hyper-{Markov} laws in the statistical analysis of decomposable graphical models},
    JOURNAL  = {Ann. Statist.},
    FJOURNAL = {The Annals of Statistics},
    VOLUME   = {21},
    YEAR     = {1993},
    NUMBER   = {3},
    PAGES    = {1272--1317},
}

@article {kim2022rydberg,
    AUTHOR   = {Kim, M. and Kim, K. and Hwang, J. and Moon, E.-G. and Ahn, J.},
    TITLE    = {Rydberg quantum wires for maximum independent set problems},
    JOURNAL  = {Nat. Phys.},
    FJOURNAL = {Nature Physics},
    VOLUME   = {18},
    YEAR     = {2022},
    NUMBER   = {7},
    PAGES    = {755--759},
}

@misc {kim2022data,
    AUTHOR       = {Kim, M. and Kim, K. and Hwang, J. and Moon, E.-G. and Ahn, J.},
    TITLE        = {Rydberg quantum wires for maximum independent set problems},
    HOWPUBLISHED = {Figshare},
    YEAR         = {2022},
    NOTE         = {Dataset and code, version 3},
}

@article {cerami2012cbioportal,
    AUTHOR   = {Cerami, E. and Gao, J. and Dogrusoz, U. and Gross, B. E. and Sumer, S. O. and Aksoy, B. A. and Jacobsen, A. and Byrne, C. J. and Heuer, M. L. and Larsson, E. and Antipin, Y. and Reva, B. and Goldberg, A. P. and Sander, C. and Schultz, N.},
    TITLE    = {The {cBio} cancer genomics portal: An open platform for exploring multidimensional cancer genomics data},
    JOURNAL  = {Cancer Discov.},
    FJOURNAL = {Cancer Discovery},
    VOLUME   = {2},
    YEAR     = {2012},
    NUMBER   = {5},
    PAGES    = {401--404},
}

@article {gao2013cbioportal,
    AUTHOR   = {Gao, J. and Aksoy, B. A. and Dogrusoz, U. and Dresdner, G. and Gross, B. and Sumer, S. O. and Sun, Y. and Jacobsen, A. and Sinha, R. and Larsson, E. and Cerami, E. and Sander, C. and Schultz, N.},
    TITLE    = {Integrative analysis of complex cancer genomics and clinical profiles using the {cBioPortal}},
    JOURNAL  = {Sci. Signal.},
    FJOURNAL = {Science Signaling},
    VOLUME   = {6},
    YEAR     = {2013},
    NUMBER   = {269},
    PAGES    = {pl1},
}

@article {ciriello2012memo,
    AUTHOR   = {Ciriello, G. and Cerami, E. and Sander, C. and Schultz, N.},
    TITLE    = {Mutual exclusivity analysis identifies oncogenic network modules},
    JOURNAL  = {Genome Res.},
    FJOURNAL = {Genome Research},
    VOLUME   = {22},
    YEAR     = {2012},
    NUMBER   = {2},
    PAGES    = {398--406},
}

@article {babur2015mutual,
    AUTHOR   = {Babur, {\"O}. and G{\"o}nen, M. and Aksoy, B. A. and Schultz, N. and Ciriello, G. and Sander, C. and Demir, E.},
    TITLE    = {Systematic identification of cancer driving signaling pathways based on mutual exclusivity of genomic alterations},
    JOURNAL  = {Genome Biol.},
    FJOURNAL = {Genome Biology},
    VOLUME   = {16},
    YEAR     = {2015},
    PAGES    = {45},
}

@incollection {leiserson2015comet,
    AUTHOR    = {Leiserson, M. D. M. and Wu, H. and Vandin, F. and Raphael, B. J.},
    TITLE     = {Co{ME}t: a statistical approach to identify combinations of mutually exclusive alterations in cancer},
    BOOKTITLE = {Research in computational molecular biology},
    SERIES    = {Lecture Notes in Comput. Sci.},
    VOLUME    = {9029},
    PAGES     = {202--204},
    PUBLISHER = {Springer, Cham},
    YEAR      = {2015},
}

@book {Lauritzen,
    AUTHOR    = {Lauritzen, S. L.},
    TITLE     = {Graphical Models},
    SERIES    = {Oxford Statistical Science Series},
    VOLUME    = {17},
    PUBLISHER = {The Clarendon Press, Oxford University Press},
    ADDRESS   = {New York},
    YEAR      = {1996},
    ISBN      = {0-19-852219-3},
}

@inproceedings {weitz_counting_independent_sets,
    AUTHOR    = {Weitz, D.},
    TITLE     = {Counting independent sets up to the tree threshold},
    BOOKTITLE = {Proceedings of the 38th Annual {ACM} Symposium on Theory of Computing ({STOC}'06)},
    PAGES     = {140--149},
    PUBLISHER = {ACM},
    ADDRESS   = {New York},
    YEAR      = {2006},
}

@article {kelly_loss_networks,
    AUTHOR   = {Kelly, F. P.},
    TITLE    = {Loss networks},
    JOURNAL  = {Ann. Appl. Probab.},
    FJOURNAL = {The Annals of Applied Probability},
    VOLUME   = {1},
    YEAR     = {1991},
    NUMBER   = {3},
    PAGES    = {319--378},
}

@article {boorstyn1987throughput,
    AUTHOR   = {Boorstyn, R. and Kershenbaum, A. and Maglaris, B. and Sahin, V.},
    TITLE    = {Throughput analysis in multihop {CSMA} packet radio networks},
    JOURNAL  = {IEEE Trans. Commun.},
    FJOURNAL = {IEEE Transactions on Communications},
    VOLUME   = {35},
    YEAR     = {1987},
    NUMBER   = {3},
    PAGES    = {267--274},
}

@article {jiang2010csma,
    AUTHOR   = {Jiang, L. and Walrand, J.},
    TITLE    = {A distributed {CSMA} algorithm for throughput and utility maximization in wireless networks},
    JOURNAL  = {IEEE/ACM Trans. Netw.},
    FJOURNAL = {IEEE/ACM Transactions on Networking},
    VOLUME   = {18},
    YEAR     = {2010},
    NUMBER   = {3},
    PAGES    = {960--972},
}

@article {rajagopalan2009network_adiabatic,
    AUTHOR   = {Rajagopalan, S. and Shah, D. and Shin, J.},
    TITLE    = {Network adiabatic theorem: An efficient randomized protocol for contention resolution},
    JOURNAL  = {SIGMETRICS Perform. Eval. Rev.},
    FJOURNAL = {ACM SIGMETRICS Performance Evaluation Review},
    VOLUME   = {37},
    YEAR     = {2009},
    NUMBER   = {1},
    PAGES    = {133--144},
}

@article {baxter_hard_hexagons,
    AUTHOR   = {Baxter, R. J.},
    TITLE    = {Hard hexagons: exact solution},
    JOURNAL  = {J. Phys. A},
    FJOURNAL = {Journal of Physics. A. Mathematical and General},
    VOLUME   = {13},
    YEAR     = {1980},
    NUMBER   = {3},
    PAGES    = {L61--L70},
}

@article {cadilhe_rsa,
    AUTHOR   = {Cadilhe, A. and Ara{\'u}jo, N. A. M. and Privman, V.},
    TITLE    = {Random sequential adsorption: From continuum to lattice and pre-patterned substrates},
    JOURNAL  = {J. Phys. Condens. Matter},
    FJOURNAL = {Journal of Physics. Condensed Matter},
    VOLUME   = {19},
    YEAR     = {2007},
    NUMBER   = {6},
    PAGES    = {065124},
}

@article {kotecky_preiss,
    AUTHOR   = {Koteck{\'y}, R. and Preiss, D.},
    TITLE    = {Cluster expansion for abstract polymer models},
    JOURNAL  = {Comm. Math. Phys.},
    FJOURNAL = {Communications in Mathematical Physics},
    VOLUME   = {103},
    YEAR     = {1986},
    NUMBER   = {3},
    PAGES    = {491--498},
}

@article {fernandez_procacci,
    AUTHOR   = {Fern{\'a}ndez, R. and Procacci, A.},
    TITLE    = {Cluster expansion for abstract polymer models: New bounds from an old approach},
    JOURNAL  = {Comm. Math. Phys.},
    FJOURNAL = {Communications in Mathematical Physics},
    VOLUME   = {274},
    YEAR     = {2007},
    NUMBER   = {1},
    PAGES    = {123--140},
}

@article {vandin2012dendrix,
    AUTHOR   = {Vandin, F. and Upfal, E. and Raphael, B. J.},
    TITLE    = {De novo discovery of mutated driver pathways in cancer},
    JOURNAL  = {Genome Res.},
    FJOURNAL = {Genome Research},
    VOLUME   = {22},
    YEAR     = {2012},
    NUMBER   = {2},
    PAGES    = {375--385},
}

@article {leiserson2013multidendrix,
    AUTHOR   = {Leiserson, M. D. M. and Blokh, D. and Sharan, R. and Raphael, B. J.},
    TITLE    = {Simultaneous identification of multiple driver pathways in cancer},
    JOURNAL  = {PLoS Comput. Biol.},
    FJOURNAL = {PLoS Computational Biology},
    VOLUME   = {9},
    YEAR     = {2013},
    NUMBER   = {5},
    PAGES    = {e1003054},
}

@article {leiserson2016wext,
    AUTHOR   = {Leiserson, M. D. M. and Reyna, M. A. and Raphael, B. J.},
    TITLE    = {A weighted exact test for mutually exclusive mutations in cancer},
    JOURNAL  = {Bioinformatics},
    FJOURNAL = {Bioinformatics},
    VOLUME   = {32},
    YEAR     = {2016},
    NUMBER   = {17},
    PAGES    = {i736--i745},
}

@article {kim2016wesme,
    AUTHOR   = {Kim, Y.-A. and Madan, S. and Przytycka, T. M.},
    TITLE    = {{WeSME}: Uncovering mutual exclusivity of cancer drivers and beyond},
    JOURNAL  = {Bioinformatics},
    FJOURNAL = {Bioinformatics},
    VOLUME   = {33},
    YEAR     = {2017},
    NUMBER   = {6},
    PAGES    = {814--821},
}

@book {matern_spatial_variation,
    AUTHOR    = {Mat{\'e}rn, B.},
    TITLE     = {Spatial Variation},
    SERIES    = {Lecture Notes in Statistics},
    VOLUME    = {36},
    EDITION   = {2nd},
    PUBLISHER = {Springer-Verlag},
    ADDRESS   = {Berlin},
    YEAR      = {1986},
}

@book {stoyan_spatial_point_patterns,
    AUTHOR    = {Chiu, S. N. and Stoyan, D. and Kendall, W. S. and Mecke, J.},
    TITLE     = {Stochastic Geometry and Its Applications},
    SERIES    = {Wiley Series in Probability and Statistics},
    EDITION   = {3rd},
    PUBLISHER = {John Wiley \& Sons},
    ADDRESS   = {Chichester},
    YEAR      = {2013},
}

@article {PeyhardiFerniqueDurand2021,
    AUTHOR     = {Peyhardi, J. and Fernique, P. and Durand, J.},
    TITLE      = {Splitting models for multivariate count data},
    JOURNAL    = {J. Multivariate Anal.},
    FJOURNAL   = {Journal of Multivariate Analysis},
    VOLUME     = {181},
    YEAR       = {2021},
    PAGES      = {Paper No. 104677, 19},
}

@article {ValiquetteEtAl2026,
    AUTHOR   = {Valiquette, S. and Peyhardi, J. and Marchand, \'{E}. and Toulemonde, G. and Mortier, F.},
    TITLE    = {Tree {P}\'{o}lya {S}plitting distributions for multivariate count data},
    JOURNAL  = {J. Multivariate Anal.},
    FJOURNAL = {Journal of Multivariate Analysis},
    VOLUME   = {211},
    YEAR     = {2026},
    PAGES    = {Paper No. 105507, 16},
}

@article {PeyhardiFernique2017,
    AUTHOR   = {Peyhardi, J. and Fernique, P.},
    TITLE    = {Characterization of convolution splitting graphical models},
    JOURNAL  = {Statist. Probab. Lett.},
    FJOURNAL = {Statistics \& Probability Letters},
    VOLUME   = {126},
    YEAR     = {2017},
    PAGES    = {59--64},
}

@book {JohnsonKotzBalakrishnan1997,
    AUTHOR    = {Johnson, N. L. and Kotz, S. and Balakrishnan, N.},
    TITLE     = {Discrete multivariate distributions},
    SERIES    = {Wiley Series in Probability and Statistics: Applied Probability and Statistics},
    NOTE      = {A Wiley-Interscience Publication},
    PUBLISHER = {John Wiley \& Sons, Inc., New York},
    YEAR      = {1997},
    PAGES     = {xxii+299},
    ISBN      = {0-471-12844-9},
}

@article {DarrochLauritzenSpeed1980,
    AUTHOR     = {Darroch, J. N. and Lauritzen, S. L. and Speed, T. P.},
    TITLE      = {Markov fields and log-linear interaction models for contingency tables},
    JOURNAL    = {Ann. Statist.},
    FJOURNAL   = {The Annals of Statistics},
    VOLUME     = {8},
    YEAR       = {1980},
    NUMBER     = {3},
    PAGES      = {522--539},
}

@article {MassamLiuDobra2009,
    AUTHOR   = {Massam, H. and Liu, J. and Dobra, A.},
    TITLE    = {A conjugate prior for discrete hierarchical log-linear models},
    JOURNAL  = {Ann. Statist.},
    FJOURNAL = {The Annals of Statistics},
    VOLUME   = {37},
    YEAR     = {2009},
    NUMBER   = {6A},
    PAGES    = {3431--3467},
}

\end{document}